\documentclass{article}
\usepackage[utf8]{inputenc}

\usepackage{amsmath,amsthm,amssymb,mathtools}  
\usepackage{fullpage}
\usepackage{authblk}
\usepackage{natbib}
\usepackage{bold-extra}
\usepackage{xcolor}

\usepackage{subcaption}

\usepackage{tikz}
\usetikzlibrary{shapes,arrows, calc,decorations.pathmorphing, decorations.pathreplacing, decorations.shapes, shapes.geometric}

\usepackage{hyperref}
\usepackage{nameref}
\usepackage[capitalise,noabbrev]{cleveref}

\usepackage{thmtools}
\usepackage{thm-restate} 
\usepackage{titlesec}
\titlespacing*{\paragraph}{0pt}{1ex plus 1ex minus .2ex}{1em}
\usepackage{graphicx}
\usepackage{epstopdf}
\usepackage{seqsplit}
\usepackage{url}
\usepackage{enumerate}

\usepackage{booktabs}
\usepackage{enumerate}

\usepackage[modulo,mathlines]{lineno}

\usepackage{csquotes}
\usepackage{tikz}

\usepackage{todonotes}

\usepackage[ruled,vlined]{algorithm2e}

\usepackage{anyfontsize}

\counterwithout{equation}{section}

\newtheorem{thm}{Theorem} 
\newtheorem{prop}{Proposition}
\newtheorem{lem}{Lemma}

\newtheorem{observation}{Observation}
\theoremstyle{definition}
\newtheorem{definition}{Definition}

\newcommand{\newchi}{\chi_{A \rightarrow B}}

\usepackage[normalem]{ulem}

\title{A strengthened bound on the number of states required to characterize maximum
parsimony distance}
\author{Mareike Fischer\footnote{University of Greifswald, Germany.}, \text{ }Steven Kelk\footnote{Maastricht University, The Netherlands. Corresponding author: \texttt{steven.kelk@maastrichtuniversity.nl}.} \text{ }and Sofia Vazquez Alferez\footnote{Utrecht University, The Netherlands.}}
\date{}

\begin{document}

\maketitle
\begin{abstract}
In this article we prove that the distance $d_{\mathrm{MP}}(T_1,T_2) = k$ between two unrooted binary phylogenetic trees $T_1, T_2$ on the same set of taxa can be defined by a character that is convex on one of $T_1, T_2$ and which has at most $2k$ states. This significantly improves upon the previous
bound of $7k-5$ states. We also show that for every $k \geq 1$ there exist two trees $T_1, T_2$ with $d_{\mathrm{MP}}(T_1,T_2) = k$ such that at least $k+1$ states are necessary in any character that achieves
this distance and which is convex on one of $T_1, T_2$. We augment these lower and
upper bounds with an empirical analysis which shows that in practice significantly fewer than $k+1$ states
are usually required.
\end{abstract}

\section{Introduction}
A phylogenetic tree is a tree whose leaves are bijectively labelled by a set $X$ of species (or more generally \emph{taxa}). Such trees are commonplace abstractions in the study of evolution: internal nodes represent putative branching events that gave rise to the set of taxa $X$.

Many methods exist for constructing phylogenetic trees. It is not uncommon for different methods to produce slightly different trees for the same set of species $X$, or a single method to give rise to multiple equally likely trees. This motivates the study of distances between phylogenetic trees. In this article we focus on a metric distance known as \emph{maximum parsimony distance} ($d_{\mathrm{MP}}$), which
quantifies the dissimilarity of two unrooted, binary phylogenetic trees on the same set of taxa $X$. The distance, whose
inception can be traced back to an idea of Bruen and Bryant \cite{bruen2008parsimony},
was first analysed in in \cite{FischerKelk2016,moulton2015parsimony}. The definition  is as follows. A character $\chi$ is a partition of $X$: each block of the character, known as a state, indicates that the taxa in the block have some evolutionary trait in common. Characters are a fundamental unit of measurement in the systematic study of evolution. The parsimony score of a tree $T$ on $\chi$ is defined as the minimum number of state changes along the edges of $T$, ranging over all extensions of the states to the interior nodes of $T$.
The maximum parsimony distance of two trees $T_1$ and $T_2$ is the maximum, ranging over all characters $\chi$, of the
absolute difference in parsimony scores of $T_1$ and $T_2$, on $\chi$.  Informally, if $T_1$ and $T_2$ have small $d_{\mathrm{MP}}$, then $T_1$ and $T_2$ will fit all characters, observed and as yet unobserved, equally well.

The importance of the maximum parsimony distance $d_{\mathrm{MP}}$ is twofold: First, unlike tree-rearrangement metrics such as the Subtree Prune and Regraft (SPR) or Tree Bisection and Reconnect (TBR) distances \cite{allen2001subtree,Hickey2008}, which are purely topological \cite{st2017shape}, $d_{\mathrm{MP}}$ captures how differently two trees explain evolutionary change, as it quantifies the largest possible discrepancy in parsimony scores across all possible characters. This makes it directly relevant to understanding the interpretive differences between competing evolutionary hypotheses. 
At the same time, it shares some of the attractive features of
SPR and TBR distance, such as being robust against the type of changes in trees that are induced by single reticulate (i.e. `horizontal') evolutionary events. In contrast, distances such as the Robinson-Foulds distance can change dramatically under the action of reticulate events \cite{allen2001subtree,SempleSteel2003}.

Computing $d_{\mathrm{MP}}$ is NP-hard \cite{FischerKelk2016,kelk2017complexity}, which has stimulated quite some research into
exactly and approximately computing the distance; we refer to \cite{jones2021maximum,DEEN2024103477} for an overview of recent algorithmic results. $d_{\mathrm{MP}}$ also turns out in practice to be a strong lower
bound on the TBR distance between two trees, $d_{\mathrm{TBR}}$, and
can greatly assist in computing this latter distance \cite{van2022reflections}, which is also NP-hard.

One crucial combinatorial and algorithmic insight is that $d_{\mathrm{MP}}$ can be computed
by restricting our attention to characters $\chi$ that are convex on one of the trees: this means that the parsimony score of that tree on $\chi$, is one less than the number of states in $\chi$. (In fact, we can further restrict our attention to convex characters that have at least two taxa per state \cite{kelk2017complexity}). 
This raises the question of just how many such convex characters there can be. In \cite{kelk2017note} it was shown that there are $\Theta(\phi^n)$ such characters, where $\phi \approx 1.618$ is the golden ratio and $n=|X|$. In 2016 Boes et al \cite{boes2016linear} adopted a more fine-grained analysis and proved that $d_{\mathrm{MP}}$ can be computed by focussing on convex characters that have at most $7k-5$ states, where $k = d_{\mathrm{MP}}$ (and $k \geq 1$). Later, as a by-product of a kernelization argument, it was shown in \cite{jones2021maximum} -- using a completely different argument -- that (ignoring additive terms) 560$k$ states are sufficient. Here we build on the results of \cite{boes2016linear} and strengthen the $7k-5$ bound to $2k$. Although we use the same machinery
as \cite{boes2016linear}, our technique and analysis is significantly more nuanced. One important feature of our analysis is
to use Fitch's algorithm \cite{Fitch71}, a well-known algorithm for computing the parsimony score of a tree on a character
$\chi$, to help us reason about the way the parsimony score of the tree changes under the action of careful changes to $\chi$. Fitch's
algorithm also plays a central role in our second main result, a lower bound, which shows that for every $k \geq 1$ a convex character with at least $k+1$ states is sometimes necessary. Clearly, there remains a gap between the $k+1$ lower bound and the $2k$ upper bound. What is the `true' bound? To develop more intuition for this we undertake in the third part of the article an empirical analysis of 644 tree pairs from an existing dataset \cite{van2022reflections}. We find that none of the tree pairs require more than $k$ states and, on average, only $0.44k$ states are required. This strengthens our belief, discussed in the conclusion, that the $2k$ upper bound can be lowered to $k+1$, although we are still some way from proving that. The conclusion also reflects on the challenges associated with closing this gap.  

Our result has \emph{indirect} algorithmic significance, in the following sense. If we have access to an upper bound $U$ on $d_{\mathrm{MP}}$, we can now compute $d_{\mathrm{MP}}$ exactly by restricting our attention to convex characters (that
have at least two taxa per state) with at most $2U$ states.  Such an upper bound can be obtained, for example, by
generating any upper bound on $d_{\mathrm{TBR}}$, and such an upper bound can be obtained by generating agreement
forests exactly or heuristically \cite{van2022reflections}. The number of convex characters grows far more slowly than the number
of general characters (which correspond to Bell numbers). Indeed, \cite{kelk2017note}, following \cite{steel1992complexity}, showed that there are (only)
$\binom{n-s-1}{s-1}$ convex characters on $s$ states (with at least two taxa per state). To put this in context: there are $\Theta(2^n)$ 2-state characters, but only a linear number of 2-state convex characters. To date the only moderately practical methods for computing (respectively, lower bounding) $d_{\mathrm{MP}}$ are based on enumeration (respectively, sampling) of convex characters \cite{kelk2017complexity,kelk2017note,van2022reflections}. Hence, the strengthening of the
upper bound from $7k-5$ to $2k$ potentially further accelerates the exact computation of $d_{\mathrm{MP}}$, when
$d_{\mathrm{MP}}$ is relatively small and good upper bounds $U$ are available. In the conclusion we also touch on the question whether good upper bounds on the number of states could contribute to a strengthening of the kernelization result from \cite{jones2021maximum}.

\section{Preliminaries}
For notation we closely follow \cite{boes2016linear}. We work exclusively with undirected graphs; in fact, trees. An unrooted binary phylogenetic $X$-tree $T$ is a tree $T =(V(T), E(T))$ with only nodes of degree $1$ (leaves) or $3$ (inner nodes) 
such that the leaves are bijectively labeled by some finite label set $X$ (where $X$ is often called the set of \emph{taxa}). 
For brevity, such a tree will often simply be called an \emph{$X$-tree} or \emph{a tree} (on $X$); if $X$ is implied it is simply omitted.

Two distinct taxa $a, b \in X$ are said to be a \emph{cherry} in tree $T$ if they have a common neighbour. A character on $X$ is a surjective map $\chi:X\rightarrow \mathcal{C}$ where $\mathcal{C}$ is a set of character \emph{states}; the number of distinct
states in the character is denoted by $|\chi|$. We can equivalently view $\chi$ as a partition of $X$ into blocks, where two taxa $x,y$ are in the same block of the partition if and only if $\chi(x) = \chi(y)$.
An \emph{extension} $\bar{\chi}$ of
a character $\chi$ to a whole $X$-tree $T$ is a map $\bar{\chi}:V(T)\rightarrow\mathcal{C}$ such that $\bar{\chi}(x) = \chi(x)$ for all $x\in X$. 
A \emph{mutation} induced by $\bar{\chi}$ in $T$ 
is an edge $\{u,v\}\in E(T)$ satisfying $\bar{\chi}(u) \neq \bar{\chi}(v)$,
and we write $\Delta(T,\bar{\chi})$ for the set of all mutation edges. 
The extension $\Bar{\chi}$ is said to be \emph{most parsimonious} if it achieves the minimum number of mutations over all possible extensions to $T$ of the character $\chi$. This leads naturally to the definition of parsimony score.

\begin{definition}
    Let $T$ be any $X$-tree and let $\chi$ be any character on $X$.

    Then the \emph{parsimony score} of $\chi$ on $T$ is $$\ell(T,\chi):=\min_{\Bar{\chi}}|\Delta(T,\Bar{\chi})| =\min_{\Bar{\chi}}\{ \{u,v\}\in E(T) \mid \Bar{\chi}(u)\neq\Bar{\chi}(v)\},$$
    where the minimum is taken over all possible extensions $\Bar{\chi}$ of the character $\chi$ to $T$.
\end{definition}

It is well-known that $\ell(T,\chi)\geq|\chi|-1$. If
$\ell(T,\chi)=|\chi|-1$
then $\chi$ is said to be a \emph{convex} character on $T$. Equivalently: $\chi$ is a convex character on $T$ if, viewed as a partition, the blocks of $\chi$ induce pairwise disjoint spanning trees in $T$.

Although characters are defined on a set $X$ of taxa, this set of taxa will often be implicit. We now use the parsimony score to define a distance function on pairs of trees.

\begin{definition}
    Let $(T_1,T_2)$ be a pair of $X$-trees.

    Then the \emph{maximum parsimony distance} between $T_1$ and $T_2$ is 
    $$d_{\mathrm{MP}}\left(T_1, T_2\right):=\max _\chi\left|\ell\left(T_1, \chi\right)-\ell\left(T_2, \chi\right)\right|,$$
    where the maximum is taken over all possible characters $\chi$ on $X$.
\end{definition}

It is known that $d_{\mathrm{MP}}$ is a metric on unrooted phylogenetic trees, hence we call it a distance \cite{FischerKelk2016}.

A character $\chi$ on a set $X$ of taxa is said to \emph{achieve distance $k$} on a pair $\left(T_1, T_2\right)$ of $X$-trees when $\left|\ell\left(T_1, \chi\right)-\ell\left(T_2, \chi\right)\right|= k$. If this character achieves distance $d_{\mathrm{MP}}\left(T_1, T_2\right)$, then we say that $\chi$ is an \emph{optimal} character for this pair of trees.

An optimal character for a pair of trees which has the additional property of being convex on at least one of the trees is called an \emph{optimal convex} character (for this pair of trees). Note that if $d_{\mathrm{MP}}\left(T_1, T_2\right) > 0$, then an optimal convex character will be convex on exactly one of $T_1$ and $T_2$.

The forest $F$ \emph{induced} by an extension $\Bar{\chi}$ (of a character $\chi$ to a $X$-tree $T$) is the forest obtained by deleting all mutation edges from $T$. Each of the connected components of $F$ is a subtree of $T$, whose nodes all share a common character state (assigned by $\Bar{\chi}$). 
We then say that two components of $F$ are \emph{adjacent} if the two corresponding subtrees of $T$ are connected by one mutation edge (they
cannot be connected by more than one mutation edge, since
there are no cycles in $T$). 
This yields a graph $G(F)$ where the nodes are the components of $F$ 
and the edges are the (unordered) pairs of adjacent components, 
which
can be identified with the mutation edges of $T$. Note that $G(F)$ is itself a tree. Often, when it is clear from the context, we will simply use $F$ to refer to $G(F)$.

When $\Bar{\chi}$ is a most parsimonious extension, each component
of the forest must contain at least one leaf of $T$. This
in turn implies that a most parsimonious extension never
introduces redundant states, i.e.,  states that were not in the
original character. Also, it is crucial to note that the forest (and
its tree structure) depends on the choice of the extension $\Bar{\chi}$: two different most parsimonious extensions may
yield different induced forests.
\begin{definition}\label{def:unique-and-repeated-states}
    Let $F$ be the forest induced by a 
    most parsimonious extension $\Bar{\chi}$ of a character $\chi$. 
    Let $C$ be the set of states used by $\Bar{\chi}$ (which will be equal to the set of states used by $\chi$). 
    We can distinguish between different kinds of
    states and components:
    \begin{itemize}
        \item a state of $\chi$ is \emph{unique} if it is assigned to exactly one component of $F$,
        \item a state of $\chi$ is \emph{repeating} if it is assigned to at least $k \geq 2$ components of $F$, in which case $k$ is the \emph{multiplicity} of the state,
        \item a component of $F$ is \emph{unique} if its assigned state is a unique state of $\chi$,
        \item a component of $F$ is \emph{repeating} if its assigned state is a repeating state of $\chi$.
    \end{itemize}
\end{definition}

Note that each state is either unique or repeating, but not
both. See Figure \ref{fig:intro} for an illustration of these concepts.

\begin{figure}[ht]      \centering\vspace{0.5cm} 
   \includegraphics[width=11cm]{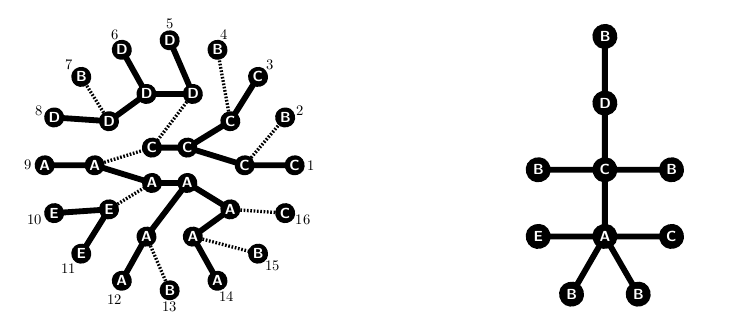} 
    \caption{Adapted from \cite{boes2016linear}. \emph{Left:} The forest $F$ induced by a most parsimonious extension $\Bar{\chi}$ of the character $\chi = (CBCBDDBDAEEABABC)$ on an $X$-tree
with leaves labeled from 1 to 16. Dotted edges are mutation edges. \emph{Right:} The corresponding graph structure $G(F)$.  States $B$ and $C$ are repeating states, while $A$, $D$ and $E$ are
unique states.
} \label{fig:intro}
  \end{figure}

Let $(T_1, T_2)$ be a pair of $X$-trees and let $\chi$ be an optimal convex character for this pair. Without loss of generality, let $\chi$ be convex on $T_1$; we emphasize that we will persist in this assumption throughout the remainder of the article.  Let $\Bar{\chi}_1$ be a most parsimonious extension
of $\chi$ to
$T_1$ and
$\Bar{\chi}_2$ a most parsimonious extension
of $\chi$ to $T_2$. Let $F_1$ and $F_2$ be the forests induced by
$\Bar{\chi}_1$ and $\Bar{\chi}_2$ 
respectively. We say that two components $A$ and $B$ are $F_i$-\emph{adjacent} if they are adjacent in the forest $F_i$.  (Note that if a state
is unique, or we are focussing on $F_1$, the term ``state'' and ``component'' can be used interchangeably.)

Let $A$ and $B$ be two distinct states that are $F_1$-adjacent. Let $\newchi$ be the new character obtained from $\chi$ by relabeling $A := B$. Then 
$\ell(T_1, \newchi) = \ell(T_1, \chi) - 1$, and $\newchi$ is also a convex character (on $T_1$) that uses exactly one fewer state than $\chi$. If $\chi$ is an optimal convex character, and $\ell(T_2, \newchi) = \ell(T_2, \chi) - 1$\footnote{Note that  $\ell(T_2, \newchi) = \ell(T_2, \chi) - 1$ is equivalent to
 $\ell(T_2, \newchi) \geq \ell(T_2, \chi) - 1$ under these circumstances. This is because if $\ell(T_2, \newchi) > \ell(T_2, \chi) - 1$, then $\newchi$ achieves a higher distance than $\chi$ i.e. the assumed optimality
 of the original character $\chi$ is contradicted.}
then $\newchi$ is also an optimal convex character \emph{with a smaller number of states}. We call such a pair of states $A$ and $B$ a \emph{good pair} (relative to $\Bar{\chi}_1$). 

Clearly, if $\chi$ is an optimal convex character for ($T_1, T_2$) that uses a minimum number of states then $\chi$ has no good pairs (relative to \emph{any} $\Bar{\chi}_1$), because otherwise this would contradict the minimality assumption. As a concrete example, consider two trees $T_1$ and $T_2$ on 4 taxa $X = \{a,b,c,d\}$ where $T_1$ consists of two cherries $\{a,b\}$ and $\{c,d\}$, and $T_2$ consists of two cherries $\{a,c\}$ and $\{b,d\}$. Here
$d_{\mathrm{MP}}(T_1,T_2) = 1$ and optimal convex characters for these two trees require at least (in fact, exactly) 2 states. A character $\chi$ which maps $a$ and $b$ to state $A$ and $c$ and $d$ to state $B$ is such a character. This is convex on $T_1$, and a most parsimonious extension  $\Bar{\chi}_1$ of $\chi$ to $T_1$ is here uniquely defined, so $F_1$ is uniquely defined. States $A$ and $B$ are $F_1$-adjacent, and $\ell(T_1, \newchi) = \ell(T_1, \chi) - 1$, but $A$ and $B$ are not a good pair (relative to $\Bar{\chi}_1$) because  $\ell(T_2, \newchi) = \ell(T_2, \chi) - 2$.
In the remainder of the article we often drop the ``relative to $\Bar{\chi}_1$'' qualification when discussing good pairs, leaving it implicit.

\paragraph{Fitch's algorithm:} In this article there is a central
role for Fitch's algorithm \cite{Fitch71}, which computes $\ell(T, \chi)$ and a corresponding most parsimonious extension $\bar{\chi}$ of $\chi$ to $T$ in polynomial time. 
Fitch's algorithm actually works on \emph{rooted} binary phylogenetic trees, so for an unrooted binary phylogenetic tree it is
first necessary to root it by arbitrarily selecting an edge and subdividing it with a single node $r$, that is known as the \emph{root} node. Subsequently all edges are directed away from $r$. This ensures that the
ancestor-descendant relationship between nodes is well-defined, and does not alter the parsimony score. The algorithm has two phases. 

In the first phase, also known as the \emph{bottom-up} phase, we start by assigning to each taxon a subset of states consisting of only the single state it is assigned by $\chi$.
The internal nodes of $T$ are assigned subsets of states recursively, as follows. Suppose
a node $p$ has two children $u$ and $v$, and the bottom-up phase has already assigned
subsets $\Phi(u)$ and $\Phi(v)$ to the two children, respectively. If $\Phi(u) \cap \Phi(v) \neq \emptyset$, then
set $\Phi(p) = \Phi(u) \cap \Phi(v)$ (in which case we say that $p$ is an \emph{intersection} node). If
$\Phi(u) \cap \Phi(v) = \emptyset$, then set $\Phi(p) = \Phi(u) \cup \Phi(v)$ (in which case we say that $p$ is a
\emph{union} node). The number of union nodes in the bottom-up phase is equal to $\ell(T,\chi)$.
To actually create a most parsimonious extension $\Bar{\chi}$ of $\chi$ to $T$, we need the second, \emph{top-down} phase of Fitch’s
algorithm\footnote{To obtain a most parsimonious extension for the original unrooted tree, we simply ignore node $r$ at the end.}. Start at the root $r$ and let $\bar{\chi}(r)$ be an \textbf{arbitrary} element in $\Phi(r)$. For an internal
node $u$ with parent $p$, we set $\bar{\chi}(u) = \bar{\chi}(p)$ (if $\bar{\chi}(p) \in \Phi(u)$) and otherwise (i.e.
$\bar{\chi}(p) \not \in \Phi(u)$) set $\bar{\chi}(u)$ to be an \textbf{arbitrary} element of $\Phi(u)$; this creates a mutation between $p$ and $u$\footnote{Note that, by construction, Fitch's algorithm never creates a mutation on \emph{both} outgoing edges of a node $p$.}. Note that the arbitrary tie-breaking choices in the top-down phase, which can arise at the root or at the head of a newly created mutation edge, can be used to generate multiple distinct most parsimonious extensions (although not necessarily \emph{all} such extensions). Later in the article we will explain how we break ties in a certain canonical way in order to discourage creation of unique states and thereby to facilitate our analysis.

Consider a \emph{rooted} binary phylogenetic tree $T$ on $X$ and a character $\chi$. Let $r$ be the root of $T$.
Suppose we run the bottom-up phase of Fitch's algorithm on $T$. The following insight is originally due to Fitch \cite{Fitch71}. Notably, the first part of the lemma applies to any/all most parsimonious extensions, not just those generated by Fitch's algorithm.

\begin{lem}
\label{lem:fitchoptimal}
A most parsimonious extension of $\chi$ to $T$ always has a state from $\Phi(r)$ at the root $r$. Also, for each
state in $\Phi(r)$, there exists a most parsimonious extension of $\chi$ to $T$ that assigns that
state to the root $r$.
\end{lem}

\noindent The lemma also extends to subtrees of $T$. Given a rooted binary phylogenetic tree $T$ on $X$, a character $\chi$ on $X$ and a node $v$ of $T$, let $T(v)$ be the subtree of $T$ rooted at $v$, and let $\chi(v)$ be the restriction of $\chi$ to the taxa at the leaves of $T(v)$. Lemma \ref{lem:fitchoptimal} implies that a most parsimonious extension of $\chi(v)$ to $T(v)$ must have a state from $\Phi(v)$ at $v$, and for each state in $\Phi(v)$ there exists a most parsimonious extension of $\chi(v)$ to $T(v)$ that assigns that
state to $v$.

\subsection*{Results}
We recall the following earlier result proven in \cite{boes2016linear}:

\begin{thm}[Bounded States Theorem]
\label{thm:original}
    Any pair $(T_1,T_2)$ of $X$-trees admits an optimal convex character with at most $7 d_{\mathrm{MP}}(T_1,T_2)-5$ states.
\end{thm}

The central idea behind this theorem is to argue that either a good pair exists, in which case a new optimal convex character can be found that uses fewer states, or that no good pairs exist and the number of states can then  be carefully bounded.

Most of the effort in the remainder of the article will be used to prove the following theorem (for a specially selected optimal character and most parsimonious extension). We prove this in Section \ref{sec:adjacency}.

\begin{restatable}[Adjacency Theorem]{theorem}{moonshot}\label{thm:moonshot}
    Let $B$ be a repeating state that occurs with multiplicity $k \geq 2$ in $F_2$. Suppose in $F_1$ that $B$ is adjacent to at least $2(k-1)$ unique states; let $U^B$ denote the set of these unique states.
    Then there exists $A \in U^B$ such that $(A,B)$ is a good pair.
\end{restatable}

We subsequently use this in Section \ref{sec:Bounding-number-of-states} to prove our main result, \cref{thm:improved-bounded-states}, which improves upon Theorem \ref{thm:original}:

\begin{restatable}[Improved Bounded States Theorem]{theorem}{improvedBoundedStates}\label{thm:improved-bounded-states}
    Any pair $(T_1,T_2)$ of $X$-trees with $d_{\mathrm{MP}}(T_1,T_2) \geq 1$ admits an optimal convex character with at most $2 d_{\mathrm{MP}}(T_1,T_2)$ states.
\end{restatable}

Next, in Section \ref{sec:mareike} we show in Proposition \ref{prop:dmpStatesNotEnough} that for every $k \geq 1$ there exists a pair of trees $T_1, T_2$ that have $d_{\mathrm{MP}}(T_1,T_2) = k$ and such that every optimal convex character for $(T_1, T_2)$ has at least $k+1$ states. Section \ref{sec:experiment} analyses these lower and upper bounds from an empirical angle. Finally, in Section \ref{sec:future} we present concluding discussions.

\section{The adjacency theorem}
\label{sec:adjacency}

In this section we assume that $\chi$ is an optimal convex character on $(T_1, T_2)$, $\chi$ is convex
on $T_1$ and that $F_1$ is the forest obtained from an arbitrary most parsimonious extension
$\bar{\chi}_1$ of $\chi$ to $T_1$. For $T_2$ we need to be more careful. We subdivide an arbitrary
edge of $T_2$ with a root $r$, and run the bottom-up phase of Fitch's algorithm on this tree. We let
$\Bar{\chi}_2$ be the most parsimonious extension of $\chi$ to $T_2$ obtained by running the top-down phase of Fitch's
algorithm with the following tie-breaking (TB) rule, and let $F_2$ be the forest obtained from $\Bar{\chi}_2$.\\
\\
    \textbf{TB rule:} When processing a node $v$ and we have an arbitrary
choice of state from $\Phi(v)$\footnote{Recall that this can only occur at the root, or at the head of a newly created mutation edge.}, if possible choose a state $s \in \Phi(v)$ that does \underline{not} cause $s$ to
become a unique state.\\ 
\\

Observe that due to the way Fitch's algorithm operates, ``$s$ becomes a unique state'' is in the context of the TB rule equivalent to: $s$ is in $\Phi(v)$,
all the taxa labelled $s$ are in the subtree $T(v$),
and the subtree of $T(v)$ induced by nodes $u$ with $s \in \Phi(u)$ is connected and includes all taxa labelled $s$.

We start by recalling two results from \cite{boes2016linear}.

\begin{observation}[Observation 3.3 in \cite{boes2016linear}]\label{obs:3.3boes2016}
 Let $A,B$ be two distinct states that are $F_1$ adjacent and such that both are unique. Then $(A, B)$ is a good pair.
\end{observation}

Given an $X$-tree $T$
and an edge $e$ of $T$, deleting $e$ breaks $T$ into two connected components and this naturally induces a bipartition $P|Q$ of $X$. We say then
that $P|Q$ is the \emph{split generated in $T$ by $e$}.

\begin{observation}[Lemma 3.2 in \cite{boes2016linear}]
\label{observation:split}
Let $A$ and $B$ be two distinct states that are $F_1$-adjacent and let $X_{A}, X_{B} \subseteq X$ be the taxa that
are labeled with $A, B$ respectively. Suppose that in $T_2$, there exists an edge $e$ that generates a split $P|Q$, where $X_A \subseteq P$ and $X_{B} \subseteq Q$.
Then $(A, B)$ is a good pair. 
\end{observation}

\moonshot*

\begin{proof}
    Recall that $F_1$ and $F_2$ can be viewed as trees where the nodes are components, and the edges are mutations. With this view in mind, let $F_2(B)$ be the unique minimal subtree of $F_2$ that spans the $B$ components in $F_2$. Every leaf of the tree $F_2(B)$ is therefore a $B$ component, but it can happen that interior nodes of $F_2(B)$ are also $B$ components. Note also that the degree of components in $F_2(B)$ can be lower than in $F_2$: this occurs when a component in $F_2(B)$ is adjacent to one or more components in $F_2$ that are not in $F_2(B)$. See Figure \ref{fig:inducedByB}.

\begin{figure}[ht]      \centering\vspace{0.5cm} 
   \includegraphics[width=3cm]{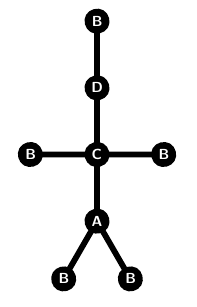} 
    \caption{The graph $F(B)$ for the forest $F$ shown in Figure \ref{fig:intro}. Note how the $E$ component and one of the $C$ components do not appear in $F(B)$, because they do not lie on the subtree of $G(F)$ that spans the $B$ components. Note also that in $F(B)$ component $A$ has degree 3, but in $G(F)$ it had degree 5. %
    }
    \label{fig:inducedByB}
  \end{figure}

As part of the proof we state and prove a number of intermediate lemmas and observations. To emphasize that these are `detours' within a larger argument we end their proofs with $\blacksquare$ rather than $\qedsymbol$.

    \begin{observation}
        If there is a state $A \in U^B$ that is not in $F_2(B)$, or which contains only one taxon, then $(A,B)$ is a good pair.
        \label{obs:sandwich}
    \end{observation}
    \begin{proof}
 \renewcommand{\qedsymbol}{$\blacksquare$}
In both cases there exists an edge in $T_2$ separating all the $A$ taxa from all the $B$ taxa. So $(A, B)$ is a good pair by Observation \ref{observation:split}.
    \end{proof}
 Observation \ref{obs:sandwich} is simple but of crucial importance. It means that, in the remainder of the proof, every state in $U^B$ is in $F_2(B)$, and contains at least two taxa (otherwise there remains nothing to show).

A simple path $P$ that starts in $F_2(B)$ at a component $A \in U^B$, ends at a $B$ component,
and which has no $B$ components as its interior nodes, is called a $B$-\emph{reaching path}.

A $B$-reaching path is
\emph{blocked} (by $C$) if there is at least one component $C\in U^B$, $C \neq A$, also on path $P$.

\begin{figure}[ht]      \centering\vspace{0.3cm} 
    \begin{tikzpicture}
\sffamily
\tikzset{every node/.style={draw, fill, circle, minimum size= 3pt, text=white}, xthick/.style ={line width = 1mm }}
    \node[] (a) at (0,0) {A};
    \node[] (c1) at (3,1) {C$_1$};
    \node[] (c2) at (3,-1) {C$_2$};
    \node[] (c3) at (-3,1) {C$_4$};
    \node[] (c4) at (-3,-1) {C$_3$};

    \path[] (c1) -- coordinate[midway](m1) (a);
    \path[] (c1) -- coordinate[midway](q11) (m1);
    \path[] (m1) -- coordinate[midway](q12) (a);
    
    \path[] (c2) -- coordinate[midway](m2) (a);
    \path[] (c2) -- coordinate[midway](q21) (m2);
    \path[] (m2) -- coordinate[midway](q22) (a);

    \path[] (c3) -- coordinate[midway](m3) (a);
    \path[] (c3) -- coordinate[midway](q31) (m3);
    \path[] (m3) -- coordinate[midway](q32) (a);

    \path[] (c4) -- coordinate[midway](m4) (a);
    \path[] (c4) -- coordinate[midway](q41) (m4);
    \path[] (m4) -- coordinate[midway](q42) (a);

    \draw[xthick] (a) -- (q12);
    \draw[xthick] (a) -- (q22);
    \draw[xthick] (a) -- (q32);
    \draw[xthick] (a) -- (q42);

    \draw[xthick, dashed] (q11) -- (q12);
    \draw[xthick, dashed] (q21) -- (q22);
    \draw[xthick, dashed] (q31) -- (q32);
    \draw[xthick, dashed] (q41) -- (q42);

    \draw[xthick] (c1) -- (q11);
    \draw[xthick] (c2) -- (q21);
    \draw[xthick] (c3) -- (q31);
    \draw[xthick] (c4) -- (q41);

    \coordinate (c1o1) at (3,1.75);
    \coordinate (c1o2) at (3.75,1);
    \coordinate (c1o3) at (3.5,1.5);
    \draw[xthick] (c1) -- (c1o1);
    \draw[xthick] (c1) -- (c1o2);
    \draw[xthick] (c1) -- (c1o3);

    \coordinate (c2o3) at (3.5,-1.5);
    \draw[xthick] (c2) -- (c2o3);

    \coordinate (c3o2) at (-3.75,1);
    \coordinate (c3o3) at (-3.5,1.5);
    \draw[xthick] (c3) -- (c3o2);
    \draw[xthick] (c3) -- (c3o3);

    \coordinate (c4o1) at (-3,-1.75);
    \coordinate (c4o2) at (-3.75,-1);
    \coordinate (c4o3) at (-3.5,-1.5);
    \coordinate (c4o4) at (-2.5,-1.5);
    \draw[xthick] (c4) -- (c4o1);
    \draw[xthick] (c4) -- (c4o2);
    \draw[xthick] (c4) -- (c4o3);
    \draw[xthick] (c4) -- (c4o4);

    \node[draw=none, fill=none, text=black, rectangle] (text1) [right of=c1o3]{Towards B};
    \node[draw=none, fill=none, text=black, rectangle] (text2) [right of=c2o3]{Towards B};
    \node[draw=none, fill=none, text=black, rectangle] (text3) [left of=c3o3]{Towards B};
    \node[draw=none, fill=none, text=black, rectangle] (text4) [left of=c4o3]{Towards B};
    
\end{tikzpicture}
    \caption{Situation ($\alpha$): unique state $A \in U^B$ is completely blocked. This means that each path in $F_2(B)$ from $A$ to a $B$ component must pass through $C_1, C_2, C_3$ or $C_4$, where $\{C_1, C_2, C_3, C_4\} \subset U^B$. The intuition is that after relabelling the $A$ taxa to $B$, the parsimony score in $T_2$ will not drop by too much due to the $C_i$ interrupting paths from the original $A$ component to the $B$ components.} \label{fig:fullyblocked}
  \end{figure}
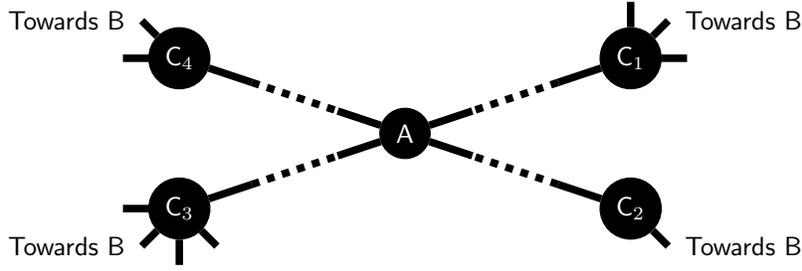

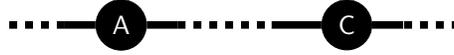
\begin{figure}[ht]      \centering\vspace{0.5cm} 
   \begin{tikzpicture}
\sffamily
\tikzset{every node/.style={draw, fill, circle, minimum size= 4pt, text=white}, xthick/.style ={line width = 1mm }}
    \node[draw, circle, minimum size=4pt] (a) at (0,0) {A};
    \node[draw, circle, minimum size=4pt] (c) at (3,0) {C};
    
    \draw[xthick] (a) -- (0.75,0);
    \draw[xthick] (-0.75,0) -- (a);
    \draw[xthick, dashed] (-0.75,0) -- (-1.5,0);
    
    \draw[xthick] (c) -- (2,0);
    \draw[xthick] (3.75,0) -- (c);
    \draw[dashed, xthick] (4.5,0) -- (3.75,0);
    
    \draw[dashed, xthick] (0.75,0) -- (2,0);
\end{tikzpicture}
    \caption{Situation ($\beta$): $A$ and $C$ are both in $U^B$, both have
    degree-2 in $F_2(B)$, and all components between them (if they exist) also have degree 2 in $F_2(B)$. Also, none of the components between $A$ and $C$ are $B$ components. The intuition is that after relabelling the $A$ taxa to $B$, the $C$ component will interrupt``rightwards'' paths from the original $A$ component to $B$ components on the right. ``Leftward'' paths from the original $A$ component to $B$ components on the left might be possible, but only one mutation can be saved in that direction due to $A$ having degree-2.} \label{fig:deg2blocked}
  \end{figure}

\begin{lem}
\label{lem:blocked}
    In $F_2(B)$, at least one of the following two situations holds:
    \begin{enumerate}
   \item[$\alpha$.] There exists $A \in U^B$ such that all $B$-reaching paths that start at $A$ are blocked. We call such an $A$ \emph{completely blocked}. See Figure \ref{fig:fullyblocked}.
    \item[$\beta$.] There exists
    distinct $A, C \in U^B$ such that $A$ and $C$ both have degree 2 in $F_2(B)$ and all components between them in $F_2(B)$ (if they exist) also have degree 2 in $F_2(B)$, and none of them are $B$ components. See Figure \ref{fig:deg2blocked}.
   \end{enumerate}
    \label{lem:sufficientlyblocked}
\end{lem}
\begin{proof}
\renewcommand{\qedsymbol}{$\blacksquare$}
We prove this by induction on the number of $B$ components $k\geq 2$ in $F_2(B)$.
For the base case, consider that for $k=2$, $F_2(B)$ consists of a path with two $B$ components, one at each end, and whereby all the components inbetween, including at least $2(k-1)=2$ distinct components $A \neq C$ from $U^B$, have degree 2 in $F_2(B)$. None of the components between $A$ and $C$ can be $B$ components, because $k=2$. Hence, we are in situation ($\beta$) and the base case holds. 

Beyond the base case, we have $k \geq 3$ and distinguish two cases: \begin{enumerate}\item[(i)]  $F_2(B)$ has no interior $B$ components; \item[(ii)] at least one interior component of $F_2(B)$ is a $B$-component. \end{enumerate}

We start with (i), i.e., only the leaves of $F_2(B)$ are $B$ components. Consider then $k \geq 3$. We assume by induction that the claim holds for all smaller $k$. First, suppose there are two distinct degree-2 components $A, C \in U^B$ 
such that only degree-2 components lie inbetween them (if any) in $F_2(B)$. The components between $A$ and $C$ cannot be $B$ components, because $B$ components are only at the leaves, so we are again in situation ($\beta$) and we are done.

Hence, we proceed by assuming that
whenever we have $A \neq C \in U^B$ such that both have degree 2,
there is at least one component on the path between them that has degree higher than 2 in $F_2(B)$; note that this cannot be a $B$ component because we are in case (i). 
Recall that a tree with $k$ leaves and no degree-2 nodes has at most $2k-3$ edges\ (of which $k$ are incident to leaves) and at most $k-2$ 
interior nodes, and both extremes are reached in the case of binary trees. Further, if we view maximal chains of degree-2 nodes as ``macro-edges'', we see that each macro-edge can have, due to the presumed non-applicability of ($\beta$), at most one component from $U^B$. Hence,
there can be at most $2k-3$ of the $U^B$ on macro-edges (all of which have degree 2).
We say that a macro-edge is a \emph{leaf} macro-edge if one of the components in it is adjacent to a leaf of $F_2(B)$ (where such a leaf will necessarily be a $B$ component). Now, if all of the $k$ leaf macro-edges
have a component from $U^B$, then let $A$ be any element of $U^B$ that is not on a leaf macro-edge. $A$ must exist, because otherwise $U^B$ contains
only $k$ components and $k < 2k-2$ for $k \geq 3$. It is easy to check that $A$ is completely blocked (by the $U^B$ components on the leaf macro-edges) and the induction claim holds. So we can assume
that at most $k-1$ components from $U^B$ are on leaf macro-edges. Combined with the observation that at most $k-2$ of the components of $U^B$ can be components of degree-3 or higher in $F_2(B)$, we see that
at most $(k-1)+(k-2)=2k-3$ components from $U^B$ can be on leaf macro-edges, or components of degree-3 or higher. There are at least $2k-2$ elements in $U^B$, so we conclude that at least one element
of $U^B$, let us call it $C$, has degree-2 and is on an \emph{interior} macro-edge (i.e. a macro-edge that is not a leaf macro-edge).

Specifically, due to being an interior macro-edge, the part of $F_2(B)$ ``left'' of $C$ contains at least two $B$ leaf components, and the part ``right'' of $C$ also contains at least two $B$ leaf components. We delete $C$, which disconnects $F_2(B)$ into two pieces, a ``left'' subtree and a ``right'' subtree; then, in each subtree, add $C$ back in its original location, followed by a new degree-1 $B$ component immediately adjacent to it (so in both subtrees $C$ has degree 2). Observe that if 
at least one of the left or right subtrees has a completely blocked component, i.e., if ($\alpha$) holds for the subtree, then so does the original $F_2(B)$. This is because, if $C$ has a blocking function in (say) the left subtree, it fulfills the same blocking function in $F_2(B)$. (Note that the completely blocked component cannot be $C$, due to the adjacency of $C$ to a $B$ component). On the other hand, suppose ($\beta$) holds for at least one of the left or right subtrees. Neither of the two degree-2 components of $U^B$ whose existence is guaranteed by ($\beta$) can be equal to $C$: if this was the case, then the macro-edge of $F_2(B)$ that contains $C$ would have already had a second degree-2 component from $U^B$ on it, with only degree-2 non-$B$ components inbetween them, contradicting the earlier assumption.
Hence, $(\beta)$ would also hold for $F_2(B)$. Summarizing: if we can show that the inductive claim holds for at least one of the subtrees, we are done for case (i). So, we need to argue that at least one of the left and right subtrees has enough components from $U^B$ in order to trigger the inductive claim. Towards this goal, let $b_l, b_r$ be the number of $B$ components in the left and right subtree, respectively, all of which are leaves. We have that $b_l + b_r = k+2$, $b_l < k$ and $b_r < k$, where the inequalities hold because, although we added one $B$ component to each of the two subtrees, each subtree lost all the at least 2 $B$ components from the other subtree.

Suppose neither the left subtree or the right subtree has enough components from $U^B$ to trigger the induction: this means that there are at most $2(b_l - 1) - 1 = 2b_l - 3$ such components in the left subtree and analogously  $2b_r-3$ in the right subtree. Given that $C$ is counted twice, there are at most $(2b_l -3) + (2b_r -3) - 1 = 2(b_l + b_r) - 7$ components in $U^B$. However, given that $b_l + b_r = k+2$, we conclude that there are at most $2(k+2)-7 = 2k-3$ elements in $U^B$: contradiction on the assumption that $U^B$ contains at least $2(k-1)$ components. Hence, recalling that  $b_l < k$ and $b_r < k$, we see that the inductive claim triggers for at least one of the left or right subtrees. This concludes the proof for the case when $F_2(B)$ has no interior $B$ components.

Next, we deal with case (ii), which is when $F_2(B)$ has at least one interior node that is a $B$ component. We call this interior
node $B'$ to disambiguate, and let $p \geq 2$ be its degree.
We split $F_2(B)$ into $p$ subtrees as follows: delete $B'$, break $F_2(B)$ into $p$ subtrees, and then add $B'$ back to each subtree in its original location. Each subtree has strictly fewer than $k$ $B$ components, because it loses at least one $B$ component from each of the other $(p-1) \geq 1$ subtrees. Similar to above, if we can show that at least one of these $p$ subtrees triggers the inductive claim we are done. That is because a component of $U^B$ that is completely blocked in a subtree is completely blocked in $F_2(B)$, and because two components of $U^B$ satisfying ($\beta$) in one of the subtrees, also satisfy it in $F_2(B)$.
To this end, let $b_1, ..., b_p$ be the number of $B$ components in the newly created subtrees. We have $\sum b_i = k + p - 1$ because $B'$ is overcounted $p-1$ times. Suppose that induction does not trigger for any of the subtrees: so each has at most $2b_i - 3$ elements from $U^B$. This would mean that there are $\sum (2b_i - 3)$ elements
in $U^B$ in total, which is $2(k + p - 1) - 3p = 2k - p - 2$. Given that $p \geq 2$, this is at most $2k-4$, yielding once again a contradiction on the original assumption that $U^B$ contained at least $2(k-1)$ components. This concludes the proof for case (ii) and thus of Lemma \ref{lem:sufficientlyblocked}.
\end{proof}

We are now ready to embark upon the final part of the proof of Theorem \ref{thm:moonshot}. The high-level idea is as follows. From Lemma \ref{lem:sufficientlyblocked} we know that at least one of situations ($\alpha$) and ($\beta$) holds. If ($\alpha$) holds, then we will argue that $(A, B)$ is a good pair, where $A$ is the
component of $U^B$ identified by ($\alpha$). The intuition behind this argument is that when the $A$ taxa are relabelled $B$, the resulting new most parsimonious extension has very little freedom to lower the parsimony score, because this would require the extension to connect the $A$ region of $T_2$ (which has now become $B$) to other $B$ components, via paths in the new most parsimonious extension that solely have state $B$. Such paths can exist, but they will necessarily ``cut through'' the components in $U^B$ that blocked the $B$-paths starting at $A$. We will argue that such a ``cutting through'' operation will necessarily induce new mutations elsewhere, neutralizing the advantage of slicing through the $U^B$ component in the first place, leading overall to a decrease of at most one in the parsimony score, i.e., that $(A, B)$ is indeed a good pair. Situation $(\beta)$ will be similar in spirit, but adapted somewhat.
For both $(\alpha)$ and $(\beta)$ the main difficulty is arguing that the ``cutting through'' operation definitely incurs a new mutation. We will leverage our special tie-breaking version of Fitch's algorithm for this.

Before continuing, recall the prevailing assumptions from the start of Section \ref{sec:adjacency}:
\begin{itemize}
    \item $\chi$ is an optimal convex character on ($T_1, T_2$);
    \item $\chi$ is convex on $T_1$;
    \item $F_1$ is the forest obtained from an arbitrary most parsimonious extension $\bar{\chi}_1$ of $\chi$ to $T_1$;
    \item $\Bar{\chi}_2$ was obtained by running the top-down phase of Fitch's algorithm with the TB rule, and $F_2$ and $F_2(B)$ were derived from this specially-selected most parsimonious extension.
    \end{itemize}

\paragraph{Leveraging Fitch's algorithm.} It is helpful at this point to view $F_2(B)$ not as components joined together by mutation edges, but rather as a subtree of $T_2$. Recall that, in order to execute Fitch's algorithm, we inserted a root $r$ into $T_2$ and all edges of $T_2$ are oriented away from the root. Note that $r$ does not necessarily lie on $F_2(B)$; it might lie elsewhere in $T_2$. Either way, it is well-defined to speak of the $F_2(B)$-\emph{root}: this is (viewing $F_2(B)$ as a subtree of $T_2$) the node of $F_2(B)$ that is closest to $r$, which will be $r$ itself if $r$ lies on $F_2(B)$. Next, for each $A \in U^B$, we define $r_A$ to be the node
 of the $A$ component that is closest to the $F_2(B)$-root.  We simply call this the \emph{root of} $A$. Informally, as we move downwards away from the $F_2(B)$-root, this is the first node that is allocated state $A$ by the most parsimonious extension we selected. The following lemma is extremely important.

\begin{lem}
\label{lem:tiebreak}
    Let $C$ be an arbitrary component from $U^B$, and $r_C$ its root. If there is at least one $B$-taxon in the subtree of $T_2$ rooted at $r_C$, and at least one $B$-taxon that is not inside this subtree, then $B \not \in \Phi(r_C)$, where as usual $\Phi(r_C)$ denotes the subset of states that Fitch's algorithm allocated node $r_C$ in its bottom-up phase when applied to $\chi$.
\end{lem}
\begin{proof}
 \renewcommand{\qedsymbol}{$\blacksquare$}
 We know that, during the top-down phase of Fitch's algorithm, $r_C$ was allocated state $C$. Now, given the precondition that
 there are $B$ taxa both within the subtree of $T_2$ rooted
 at $r_C$, and outside it, $r_C \neq r$. Hence, $r_C$ has
 a parent $p$. Node $p$ cannot be allocated state
 $C$ by the extension, because then $p$ or some node even closer to the root would have been the
 root of $C$, not $r_C$ (contradiction). Hence, the edge $(p,r_C)$ is a mutation edge. This means that when executing the top-down phase of Fitch's algorithm, none of the states in $\Phi(r_C)$ were equal to the state allocated to $p$, so we selected a state for $r_C$ from
 $\Phi(r_C)$, using our tie-breaking rule: we selected state $C$. Suppose, for the sake of contradiction, $\Phi(r_C)$ also contained $B$. Noting that (i) $p$ was definitely not allocated state $B$ (if it had been, we would definitely have allocated state $B$ to $r_C$ instead of $C$) and (ii) the existence by assumption of $B$ taxa on both sides of the split induced by the mutation edge $(p,r_C)$, we observe that at the point the state for $r_C$ was selected, $B$ could definitely not be (or become) a unique state. Hence, the tie-breaking rule would never have picked state $C$ in the first place (it would have preferred $B$ or some other repeating state): contradiction.
\end{proof}

\begin{lem}
\label{lem:aboveAndBelow}
    Let $C$ be an arbitrary component from $U^B$, and $r_C$ its root. If there is at least one $B$-taxon in the subtree of $T_2$ rooted at $r_C$, and at least one $B$-taxon that is not inside this subtree, then no most parsimonious extension of (the restriction of) $\chi$ to the subtree $T_2(r_C)$ 
    allocates state $B$ to $r_C$. On the other hand, there is at least one most parsimonious extension of (the restriction of) $\chi$ to the subtree $T_2(r_C)$ that allocates state $C$ to $r_C$.
\end{lem}
\begin{proof}
 \renewcommand{\qedsymbol}{$\blacksquare$}
The first statement immediately follows by combining Lemma \ref{lem:fitchoptimal}
 and Lemma \ref{lem:tiebreak}. The second statement holds because our execution of Fitch's algorithm allocated state $C$ to $r_C$, and (again from Lemma \ref{lem:fitchoptimal}) the top-down phase of Fitch's algorithm only assigns states to nodes that allow a most-parsimonious extension for the subtree rooted at that node.
\end{proof}

Informally, the previous two lemmas say: if we could have allocated state $B$ (or some other repeating state) to $r_C$ without losing optimality, we would have done this already instead of selecting state $C$. So, assigning state $C$ to $r_C$ was in some sense truly necessary to obtain local optimality (i.e., in the subtree rooted at $r_C$).

\begin{lem}
\label{lem:easysplit}
    Let $A, C$ be distinct components from $U^B$, and $r_A, r_C$ their roots in $T_2$. If there is a directed path from $r_A$ to $r_C$, 
    and $B \in \Phi(r_C)$, then $(A,B)$ is a good pair.
\end{lem}
\begin{proof}
 \renewcommand{\qedsymbol}{$\blacksquare$}
 From Lemma \ref{lem:fitchoptimal}, and the assumption of the current lemma, there exists a most
 parsimonious extension of (the restriction of) $\chi$ to the subtree $T_2(r_C)$ which
 assigns state $B$ to $r_C$. From this, it follows that at least
 one taxon in the subtree $T_2(r_C)$ is assigned state $B$ by
 $\chi$. If this was not the case, then by construction the bottom-up phase of Fitch's algorithm would never have put state $B$
 in $\Phi(r_C)$ in the first place. Now, due to Lemma \ref{lem:aboveAndBelow}, there cannot be any taxa outside $T_2(r_C)$ that are assigned state $B$ by $\chi$ (because this would contradict $B \in \Phi(r_C)$). Let
 $(p,r_C)$ be the edge entering $r_C$ in $T_2$. This edge defines a split
 of $T_2$ with all the taxa labelled $A$ on one side, and all the taxa labelled
 $B$ on the other. Hence, from Observation \ref{observation:split} $(A,B)$ is a good pair.
\end{proof}
An immediate consequence of the above lemma is that whenever there is a path from $r_A$ to $r_C$, then either $(A,B)$ is a good pair, or $B \not \in \Phi(r_C)$. 

\begin{figure}[ht]      \centering\vspace{0.5cm}
\begin{tikzpicture}
\sffamily
\tikzset{every node/.style={draw, circle, minimum size= 5pt, text=black}, 
xthick/.style ={line width = 1mm },
line around/.style={decoration={pre length=#1,post length=#1}},
pathB/.style={ultra thick,blue!50,decorate,decoration={shape backgrounds,shape=circle,shape size=0.2mm,shape sep = 0.1mm, raise = 2.5pt}}
}
\tikzset{
  middle dotted line/.style={
    decoration={show path construction, 
      lineto code={
          \draw[xthick,#1] (\tikzinputsegmentfirst) --($(\tikzinputsegmentfirst)!.25!(\tikzinputsegmentlast)$);,
          \draw[dotted,xthick,#1] ($(\tikzinputsegmentfirst)!.25!(\tikzinputsegmentlast)$)--($(\tikzinputsegmentfirst)!.75!(\tikzinputsegmentlast)$);,
          \draw[xthick,#1] ($(\tikzinputsegmentfirst)!.75!(\tikzinputsegmentlast)$)--(\tikzinputsegmentlast);,
      }
    },
    decorate
  },
}
 \coordinate (center) at (0,0);

    \node[minimum size = 60pt, fill = none] (a) at (0,0) {};
    \coordinate[above right=3cm of center] (c1);
    \node[below right=1.5cm and 2cm of center] (c2) {C$_1$};
    \node[below left =2.5cm of center] (c4) {C$_3$};
    \node[below=1.5cm and 2cm of center] (c3){C$_2$};


    \node[fill, label={[label distance = 0mm]-5:{$r_{C_3}$}}] (rc4) at (c4.north east) {}; 
     \node[fill, label={[label distance = 0mm]-20:{$r_{C_2}$}}] (rc3) at (c3.north east) {}; 
     \node[fill, label={[label distance = 0mm]20:{$r_{C_1}$}}] (rc2) at (c2.north west) {}; 
     \node[fill, label={[label distance = 0mm]120:{$r$}}] (rc1) at (c1.south west) {};

    \coordinate (lc4) at (a.south west); 
    \coordinate (lc3) at (a.south); 
     \coordinate (lc2) at (a.south east);
     \node[fill, label={[label distance = 0mm]-20:{$r_{A}$}}] (lc1) at (a.north east) {};

    \node[fill, above right=0.25cm and 0.25cm of a, label={[label distance = 0mm]-20:{$p$}}] (p) {};
    
    \coordinate[below right=0.8cm and 0.8cm of c1](off-root);
    
     \coordinate (split) at ($(center)!0.5!(lc2)$);
    \node[fill=none, draw=none, rotate=45, above =1mm of split](lbA){B (was A)};
     \draw[] (lc1)-- (center);
     \draw[] (center)-- (lc2);
     \draw[] (center)-- (split);
     \draw[] (split)-- (lc3);
     \draw[] (center)-- (lc4);
    \draw[xthick](p)--(lc1);
    \draw[xthick, dotted] (p)--(rc1);
    \draw[xthick,middle dotted line] (lc2)--(rc2);
    \draw[xthick,middle dotted line] (lc3)--(rc3);
    \draw[xthick,middle dotted line] (lc4)--(rc4);

    \draw[xthick, dotted] (rc1)--(off-root);

    \coordinate[below right= 0.25cm and 0.25cm of c3] (c3o3);

    \draw[xthick] (c3) -- (c3o3);

    \coordinate[above right= 0cm and 0.5cm of c2](c2o1);
    \coordinate[right= 0cm and 0.5cm of c2] (c2o2);
    \coordinate[below right= 0cm and 0.5cm of c2] (c2o3);
    
    \draw[xthick] (c2) -- (c2o1);
    \draw[xthick] (c2) -- (c2o2);
    \draw[xthick] (c2) -- (c2o3);

 

    \coordinate[below right= 0.25cm and 0.25cm of c4] (c4o1);
    \coordinate[below= 0.25cm and 0.25cm of c4] (c4o2);
    \coordinate[below left= 0.25cm and 0.25cm of c4] (c4o3);
    \coordinate[left= 0.25cm and 0.25cm of c4] (c4o4);
    
    \draw[xthick] (c4) -- (c4o1);
    \draw[xthick] (c4) -- (c4o2);
    \draw[xthick] (c4) -- (c4o3);
    \draw[xthick] (c4) -- (c4o4);

    \node[draw=none, fill=none, text=black, rectangle] (text2) [right of=c2o2]{Towards B};
    \node[draw=none, fill=none, text=black, rectangle] (text3) [below= 1mm of c3o3]{Towards B};
    \node[draw=none, fill=none, text=black, rectangle] (text4) [below= 1mm of c4o2]{Towards B};

    \draw[pathB]($(rc1)!0.5!(p)$) -- (center);
    \draw[pathB](center) -- (lc4);
    \draw[pathB](lc4) -- ($(lc4)!0.7!(rc4)$);
    \draw[pathB](center) -- (split);
    \draw[pathB](split) -- (lc2);
    \draw[pathB](lc2) -- ($(lc2)!0.7!(rc2)$);
    \draw[pathB](split) -- (lc3);
    \draw[pathB](lc3) -- ($(lc3)!0.5!(rc3)$);

    \node[rectangle, draw=none, above left= of center, text width=5cm, text=black!50](edge){When switching back to A, do not relabel above this edge (this creates at most one mutation)};
    \node[rectangle, draw= none, right= of rc1, text width=5cm, text=black!50](lbRoot){Root $r$ (to make Fitch's algorithm possible).};

    \draw[->,dotted, thick, black!50]([xshift=-0.5cm]edge.east)--([yshift=0.1cm]$(p)!0.5!(lc1)$);
\end{tikzpicture}
    \caption{This figure continues from Figure \ref{fig:fullyblocked}, and illustrates Lemma \ref{lem:alphagood}. There are directed paths from $r_A$ to $r_{C_1},r_{C_2}$ and $r_{C_3}$. This means that, after relabelling $A$ to $B$, and constructing a new most parsimonious extension, there cannot be any nodes labelled $B$ - here $B$ is shown in blue - that pass through 
    $r_{C_1},r_{C_2}$ or $r_{C_3}$. This is a consequence of the way we broke ties in the top-down phase of Fitch's algorithm. This means that,  when switching back to state $A$, there is at most one new mutation induced: on the edge between $r_A$ and its parent $p$.}\label{fig:relabel1}
  \end{figure}
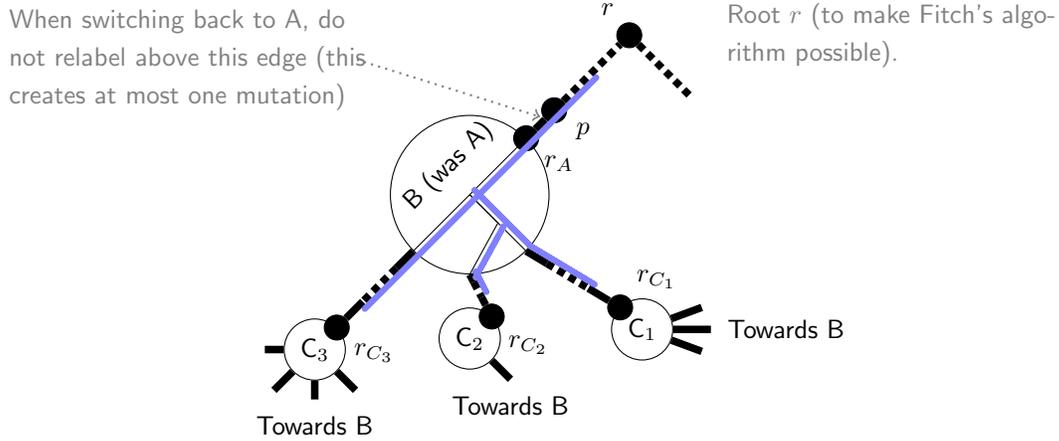

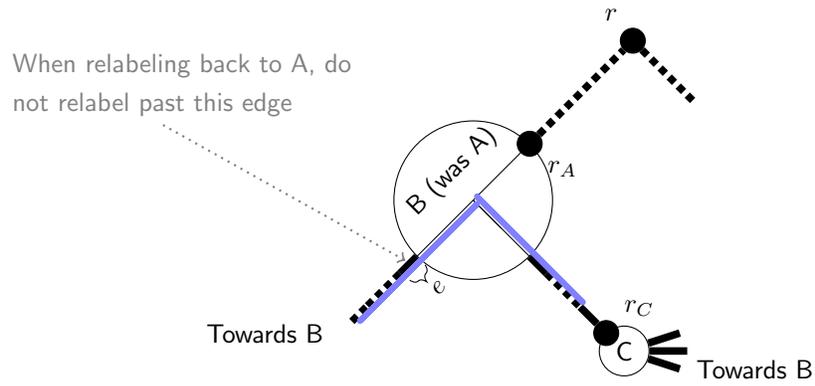
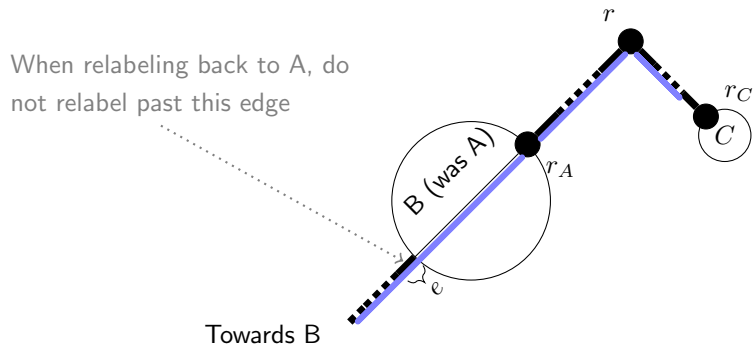
\begin{figure}[ht]
\centering
   \begin{subfigure}[b]{\textwidth}
        \centering
       \begin{tikzpicture}
\sffamily
\tikzset{every node/.style={draw, circle, minimum size= 5pt, text=black}, 
xthick/.style ={line width = 1mm },
line around/.style={decoration={pre length=#1,post length=#1}},
pathB/.style={ultra thick,blue!50,decorate,decoration={shape backgrounds,shape=circle,shape size=0.2mm,shape sep = 0.1mm, raise = 2pt}}
}
\tikzset{
  middle dotted line/.style={
    decoration={show path construction, 
      lineto code={
          \draw[xthick,#1] (\tikzinputsegmentfirst) --($(\tikzinputsegmentfirst)!.25!(\tikzinputsegmentlast)$);,
          \draw[dotted,xthick,#1] ($(\tikzinputsegmentfirst)!.25!(\tikzinputsegmentlast)$)--($(\tikzinputsegmentfirst)!.75!(\tikzinputsegmentlast)$);,
          \draw[xthick,#1] ($(\tikzinputsegmentfirst)!.75!(\tikzinputsegmentlast)$)--(\tikzinputsegmentlast);,
      }
    },
    decorate
  },
}

\tikzset{
  trailoff line/.style={
    decoration={show path construction, 
      lineto code={
          \draw[xthick,#1] (\tikzinputsegmentfirst) --($(\tikzinputsegmentfirst)!.25!(\tikzinputsegmentlast)$);,
          \draw[dotted,xthick,#1] ($(\tikzinputsegmentfirst)!.25!(\tikzinputsegmentlast)$)--($(\tikzinputsegmentfirst)!.75!(\tikzinputsegmentlast)$);,
          \draw[dotted,xthick,#1] ($(\tikzinputsegmentfirst)!.75!(\tikzinputsegmentlast)$)--(\tikzinputsegmentlast);,
      }
    },
    decorate
  },
}
 \coordinate (center) at (0,0);

    \node[minimum size = 60pt, fill = none] (a) at (0,0) {};
    \coordinate[above right=3cm of center] (r);
    \node[below right=2.5cm of center] (c2) {C};
     \node[fill, label={[label distance = 0mm]20:{$r_{C}$}}] (rc2) at (c2.north west) {}; 
     \node[fill, label={[label distance = 0mm]120:{$r$}}] (rr) at (r.south west) {};

    \coordinate (lc4) at (a.south west); 
    \coordinate (lc3) at (a.south); 
     \coordinate (lc2) at (a.south east);
     \node[fill, label={[label distance = 0mm]-20:{$r_{A}$}}] (lr) at (a.north east) {};

    \coordinate[below right=0.8cm and 0.8cm of r](off-root);
    \coordinate[below left =2.5cm of center] (c4);
    
     \coordinate (split) at ($(center)!0.5!(lc2)$);
    \node[fill=none, draw=none, rotate=45, above =1mm of split](lbA){B (was A)};
     \draw[] (lr)-- (center);
     \draw[] (center)-- (lc2);
     \draw[] (center)-- (split);
     \draw[] (center)-- (lc4);
    \draw[xthick, dotted] (lr)--(rr);
    \draw[xthick,middle dotted line] (lc2)--(rc2);
    \draw[xthick,trailoff line] (lc4)--(rc4);

    \draw[xthick, dotted] (rr)--(off-root);

    \coordinate[above right= 0cm and 0.5cm of c2](c2o1);
    \coordinate[right= 0cm and 0.5cm of c2] (c2o2);
    \coordinate[below right= 0cm and 0.5cm of c2] (c2o3);
    
    \draw[xthick] (c2) -- (c2o1);
    \draw[xthick] (c2) -- (c2o2);
    \draw[xthick] (c2) -- (c2o3);

    \node[draw=none, fill=none, text=black, rectangle] (text2) [right of=c2o3]{Towards B};
    
    \node[draw=none, fill=none, text=black, rectangle] (text2) [left of=c4]{Towards B};

    \draw[pathB](center) -- (lc4);
    \draw[pathB](lc4) -- ($(lc4)!0.8!(c4)$);
    \draw[pathB](center) -- (lc2);
    \draw[pathB](lc2) -- ($(lc2)!0.65!(rc2)$);

    \node[rectangle, draw=none, above left= of center, text width=5cm, text=black!50](edge){When relabeling back to A, do not relabel past this edge};

    \draw[->,dotted, thick, black!50]([xshift=-0.5cm]edge.south)--([yshift=0.1cm]$(lc4)!0.15!(rc4)$);

    \coordinate[] (ref) at ($(lc4)!0.2!(c4)$);
    \node[draw=none, rotate= 45,fill=none, sloped, below right=3mm of ref](ce) {$e$};

    \draw [decorate,decoration={brace,amplitude=5pt,raise=1ex}]
  (lc4) -- ($(lc4)!0.2!(c4)$);
\end{tikzpicture}
       \subcaption{Case (1)}
   \end{subfigure}
   
   \begin{subfigure}[b]{\textwidth}
        \centering
       \begin{tikzpicture}
\sffamily
\tikzset{every node/.style={draw, circle, minimum size= 5pt, text=black}, 
xthick/.style ={line width = 1mm },
line around/.style={decoration={pre length=#1,post length=#1}},
pathB/.style={ultra thick,blue!50,decorate,decoration={shape backgrounds,shape=circle,shape size=0.2mm,shape sep = 0.1mm, raise = 2pt}}
}
\tikzset{
  middle dotted line/.style={
    decoration={show path construction, 
      lineto code={
          \draw[xthick,#1] (\tikzinputsegmentfirst) --($(\tikzinputsegmentfirst)!.25!(\tikzinputsegmentlast)$);,
          \draw[dotted,xthick,#1] ($(\tikzinputsegmentfirst)!.25!(\tikzinputsegmentlast)$)--($(\tikzinputsegmentfirst)!.75!(\tikzinputsegmentlast)$);,
          \draw[xthick,#1] ($(\tikzinputsegmentfirst)!.75!(\tikzinputsegmentlast)$)--(\tikzinputsegmentlast);,
      }
    },
    decorate
  },
}

\tikzset{
  trailoff line/.style={
    decoration={show path construction, 
      lineto code={
          \draw[xthick,#1] (\tikzinputsegmentfirst) --($(\tikzinputsegmentfirst)!.25!(\tikzinputsegmentlast)$);,
          \draw[dotted,xthick,#1] ($(\tikzinputsegmentfirst)!.25!(\tikzinputsegmentlast)$)--($(\tikzinputsegmentfirst)!.75!(\tikzinputsegmentlast)$);,
          \draw[dotted,xthick,#1] ($(\tikzinputsegmentfirst)!.75!(\tikzinputsegmentlast)$)--(\tikzinputsegmentlast);,
      }
    },
    decorate
  },
}
 \coordinate (center) at (0,0);

    \node[minimum size = 60pt, fill = none] (a) at (0,0) {};
    \coordinate[above right=3cm of center] (r);
    
    \node[below right=1cm and 1cm of r](off-root){$C$};
     \node[fill, label={[label distance = 0mm]120:{$r$}}] (rr) at (r.south west) {};
     \node[fill, label={[label distance = 0mm]10:{$r_{C}$}}] (rc) at (off-root.north west) {};

    \coordinate (lc4) at (a.south west); 
    \coordinate (lc3) at (a.south); 
    \node[fill, label={[label distance = 0mm]-20:{$r_{A}$}}] (lr) at (a.north east) {};


    \coordinate[below left =2.5cm of center] (c4);
    
     \coordinate (split) at ($(center)!0.5!(lc2)$);
    \node[fill=none, draw=none, rotate=45, above =1mm of split](lbA){B (was A)};
     \draw[] (lr)-- (center);
     \draw[] (center)-- (lc4);
    \draw[middle dotted line] (lr)--(rr);
    \draw[xthick,trailoff line] (lc4)--(rc4);

    \draw[middle dotted line] (rr)--(rc);

    \node[draw=none, fill=none, text=black, rectangle] (text2) [left of=c4]{Towards B};

    \draw[pathB](center) -- (lc4);
    \draw[pathB](lc4) -- ($(lc4)!0.8!(c4)$);
    \draw[pathB](lc1) -- (center);
    \draw[pathB](rr) -- (lc1);
    \draw[pathB]($(rc)!0.3!(rr)$) -- (rr);

    \node[rectangle, draw=none, above left= of center, text width=5cm, text=black!50](edge){When relabeling back to A, do not relabel past this edge};

    \draw[->,dotted, thick, black!50]([xshift=-0.5cm]edge.south)--([yshift=0.1cm]$(lc4)!0.15!(rc4)$);

    \coordinate[] (ref) at ($(lc4)!0.2!(c4)$);
    \node[draw=none, rotate= 45,fill=none, sloped, below right=3mm of ref](ce) {$e$};

    \draw [decorate,decoration={brace,amplitude=5pt,raise=1ex}]
  (lc4) -- ($(lc4)!0.2!(c4)$);
\end{tikzpicture}
       \subcaption{Case (2)}
   \end{subfigure}
    \caption{This figure illustrates the two subcases of Lemma \ref{lem:betagood}, and continues from Figure \ref{fig:deg2blocked}. In case (1) one of $A$ and $C$ is an ancestor of the other; in case (2) they are incomparable. Note that in both cases the root $r_C$ of $C$ will not be allocated state $B$ (shown as blue) in the new most parsimonious extension. When switching back from $B$ to $A$ the only place that a new mutation can be created is on edge $e$. }\label{fig:relabel2}
  \end{figure}

\begin{lem}
\label{lem:alphagood}
Suppose ($\alpha$) of Lemma \ref{lem:blocked} holds. Let $A$ be the completely blocked component in $U^B$ whose existence is guaranteed by ($\alpha$). Then
$(A,B)$ is a good pair.
\end{lem}
\begin{proof}
     \renewcommand{\qedsymbol}{$\blacksquare$}
See Figure \ref{fig:relabel1} for a schematic depiction of the ideas in this proof.

Consider any component $C \in U^B$ such that
there is a directed path from $r_A$ to $r_C$\footnote{Such a component $C$ will definitely exist, due to $A$ being a
node of $F_2(B)$ and being completely blocked, but we do not actually need this fact.}.

If $B \in \Phi(r_C)$, then by Lemma \ref{lem:easysplit} $(A,B)$ is a good pair and we are immediately done. Hence, for all such $C$ we have that $B \not \in \Phi(r_C)$. Now, let $X_A$ be the set of all taxa that are assigned $A$ by $\chi$. 
 Due to the fact that $A$ is completely blocked, any (undirected) path in $T_2$ from a taxon in $X_A$ to a taxon labelled $B$, must pass through $r_C$ -- where $C$ is some state in $U^B$ that blocked $A$ and there is a directed path from $r_A$ to $r_C$ -- or through the edge $(p, r_A)$, where $p$ is the parent of $r_A$\footnote{Due to $A$ being completely blocked, a path that passes through the edge $(p,r_A)$ will also pass through some state $C \in U^B$, but we cannot guarantee that it will pass through $r_C$, due to $C$ not being a descendant of $A$. This is why we treat this case differently.}.

We create a new character from $\chi$, called $\newchi$, in which the taxa labelled $A$ by $\chi$ are now labelled $B$.
Due to the adjacency of $A$ and $B$ in $F_1$, $\ell(T_1, \newchi) = \ell(T_1, \chi)-1$, and $\newchi$ has one fewer state than $\chi$, so it is convex on $T_1$.
Our goal is to show that $\ell(T_2, \newchi) \geq \ell(T_2, \chi)-1$. Assume that this does not hold, i.e., that
$\ell(T_2, \newchi) \leq  \ell(T_2, \chi)-2$. We will derive a contradiction on this assumption. We start by constructing a most parsimonious extension of $\newchi$ to $T_2$; for this purpose we use Fitch's algorithm with the same root location as before (but this time arbitrary tie-breaking in the top-down phase is sufficient). Observe that if, in the original most parsimonious extension, $r_C$ could be reached by a directed path from $r_A$, then the set of states allocated by the bottom-up phase of Fitch's algorithm to node $r_C$, will be \emph{the same} for $\chi$ and $\newchi$. This is because the subtree $T_2(r_C)$ does not contain any taxa from $X_A$, and hence the restriction of $\chi$ to the taxa in $T_2(r_C)$ is equal to the restriction of $\newchi$ to this subtree; and thus, the bottom-up phase of Fitch's algorithm will execute identically for $\chi$ and $\newchi$ on this subtree, assigning the subset $\Phi(r_C)$ to $r_C$ in both
cases. Now, we know that $B \not \in \Phi(r_C)$. 
Hence, the most parsimonious extension constructed by Fitch's algorithm for $\newchi$ will not assign $B$ to $r_C$. This holds for every
 state $C \in U^B$ which blocked $A$ and had the property that there was a directed path from $r_A$ to $r_C$ in the original most parsimonious extension.

As a consequence of this, we have a most parsimonious extension of $\newchi$ to $T_2$ with the following property: any (undirected) path that starts at a taxon in $X_A$ and reaches a taxon with state $B$ such that all nodes on the path have state $B$, \emph{must} pass through the edge $(p, r_A)$ that feeds into $r_A$. Now, we switch the taxa $X_A$ back to state $A$, and any adjacent interior $B$ nodes also to $A$, and so on, continuing this
\textbf{flooding} operation to saturation with one exception: in the `upwards' direction we do not go beyond $r_A$. This might mean a single new mutation is incurred on edge $(p, r_A)$, but nowhere else. 
Hence, we have an extension of $\chi$ to $T_2$ that has at most
$\ell(T_2, \newchi)+1$ mutations. Given the assumption that $\ell(T_2, \newchi) \leq \ell(T_2, \chi)-2$, this gives
$\ell(T_2, \chi) \leq \ell(T_2, \chi)-1$: contradiction. So, indeed, $(A,B)$ is a good pair.
\end{proof}

\begin{lem}
\label{lem:betagood}
Suppose ($\beta$) of Lemma \ref{lem:blocked} holds. Let $A, C$ be the two distinct components in $U^B$ whose existence is guaranteed by ($\beta$). Then at
least one of $(A,B)$, $(C,B)$ is a good pair.
\end{lem}
\begin{proof}
     \renewcommand{\qedsymbol}{$\blacksquare$}
     See Figure \ref{fig:relabel2}.
     We distinguish two cases.
     \begin{enumerate}
         \item There is a directed path in $T_2$ from $r_A$ to $r_C$, or from $r_C$ to $r_A$. Suppose without loss of generality that the path is from $r_A$ to $r_C$. From Lemma \ref{lem:easysplit}, we either have that $(A,B)$ is a good pair, in which we are done, or $B \not \in \Phi(r_C)$. We define $X_A$ and $\newchi$ as in the previous lemma, and construct (using Fitch's algorithm with the same root location) a most parsimonious extension for $\newchi$. As argued in the previous lemma, $\Phi(r_C)$ is unchanged in this new extension, so $r_C$ is also not assigned $B$ in this new extension. Hence, any path starting from a taxon in $X_A$ and reaching a $B$ taxon such that all nodes of the path apart from the first one have state $B$, cannot pass through $r_C$. Moreover, due to the fact that $A$ has degree-2 in $F_2(B)$, any such paths must pass through a single edge $e$ (corresponding to the unique mutation edge incident to $A$ that led towards $B$ components, but \emph{away} from $r_C$). The flooding operation described at the end of the previous lemma works here too, with the difference that we should not cross edge $e$ when switching $B$ nodes back to $A$. As before, this incurs at most one new mutation in the transformation from the most parsimonious extension for $\newchi$ to an extension for
         $\chi$, proving that $\ell(T_2, \newchi) \geq \ell(T_2, \chi)-1$, and thus proving that $(A,B)$ is indeed a good pair.
         \item The previous case does not hold i.e., $r_A$ and $r_C$ are incomparable. This means that $r_A$ and $r_C$ are on opposite sides of the root $r$ in $T_2$.  Suppose there are no $B$ taxa in $T_2(r_A)$; then $(A,B)$ is a good pair because the edge entering $r_A$ induces
         a split separating the $A$ and $B$ taxa in $T_2$ completely, and we are done. Symmetrically, if there
         are no $B$ taxa in $T_2(r_C)$, then $(C,B)$ is a good pair, and we are done. So, we assume that there are $B$ taxa in both $T_2(r_A)$ and $T_2(r_C)$.
         Hence, by Lemma \ref{lem:aboveAndBelow}, $B \not \in \Phi(r_A)$ and $B \not \in \Phi(r_C)$. We 
         show that $(A,B)$ is a good pair. (In fact, so is $(C,B)$ but it is not necessary to argue this). Consider that all paths that start at a taxon in $X_A$ and reach a $B$ taxon along nodes that also have state $B$, while avoiding $r_C$, must pass through the same mutation edge $e$ incident to $A$. This is, as in the previous case, because $A$ has degree-2; in this case $e$ will lead downwards towards the $B$ taxa that are reachable by directed paths from $r_A$. Hence, just like in the previous case, applying the flooding operation from taxa in $X_A$ but not crossing edge $e$ transforms a most parsimonious extension for $\newchi$ into an extension for $\chi$, whilst inducing at most one new extra mutation.
     \end{enumerate}

\end{proof}

 The proof of Theorem \ref{thm:moonshot} is concluded by observing that by Lemma \ref{lem:blocked} at least one of ($\alpha$), ($\beta$) holds, and thus by Lemmas \ref{lem:alphagood} and \ref{lem:betagood}, respectively, a good pair certainly exists.
\end{proof}

\section{Bounding the number of states}\label{sec:Bounding-number-of-states}

\improvedBoundedStates*
\begin{proof}
 Let $\chi$ be an optimal convex character that has a minimum number of states. We construct most parsimonious extensions
$\Bar{\chi}_1$ and $\Bar{\chi}_2$ of $\chi$ to $T_1$ and $T_2$ respectively (and the resulting forests $F_1$ and $F_2$) in the same way as described at the start of Section \ref{sec:adjacency}.
By assumption $\chi$ has no good pairs, because a good pair by definition would allow the number of states to be further reduced, contradicting minimality. Hence, by Observation \ref{obs:3.3boes2016} there are no adjacencies between unique states in $F_1$. Moreover, by Theorem \ref{thm:moonshot} and the assumed absence of good pairs, a repeated state that occurs with multiplicity $k \geq 2$ in $F_2$ has at most $2(k-1)-1 = 2k-3$ unique neighbours in $F_1$.

 Let $r$ be the number of repeating states and $u$ the number of unique states. Now, let $r_k$ be the number of repeating states that have multiplicity $k \geq 2$. Given that $\left(\sum_{k \geq 2} r_k\right) = r$, the total number $t$ of states is
 \begin{equation}
 \label{eq:total}
 t=\left(\sum_{k \geq 2} r_k\right) + u.
 \end{equation}

Given that $F_1$ is connected, and that no two unique states are adjacent in $F_1$, each unique state must therefore have at least one repeating state neighbour in $F_1$. Combining this with Theorem \ref{thm:moonshot} gives,
\begin{equation}
\label{eq:bound}
u \leq \sum_{k \geq 2} (2k-3) r_k.
\end{equation}
Hence from (\ref{eq:total}) and (\ref{eq:bound}) for the total number $t$ 
of states we have: 
\begin{align}
&  t \leq \left(\sum_{k \geq 2} r_k\right) + \sum_{k \geq 2} (2k-3) r_k = \sum_{k \geq 2} (2k-2) r_k = 2 \sum_{k \geq 2} (k-1) r_k \label{eq:prior}.
\end{align}
Next, we will show the following:
\begin{equation}
\label{eq:crucial}
\sum_{k \geq 2} (k-1)r_k = d_{\mathrm{MP}}(T_1,T_2).
\end{equation}
To see that this it holds, recall
that the parsimony score of $T_1$ on $\chi$ is $|\chi|-1$ (due to convexity) and that
the parsimony score on $T_2$ is thus $|\chi|-1 + d_{\mathrm{MP}}(T_1,T_2)$ (due to $\chi$ being an optimal
character). Now, the parsimony score on $T_2$ is also equal to the number of
components in $F_2$ minus 1, which by definition of $r_k$ is in turn equal to
 \[
 \left(\sum_{k \geq 2} k \cdot r_k\right) + u - 1,
 \]
so we have
\[
 \left(\sum_{k \geq 2} k \cdot r_k\right) + u - 1 = |\chi|-1 + d_{\mathrm{MP}}(T_1,T_2).
\]

We also know that $|\chi| = r + u$, so substituting $u = |\chi|-r$ into the above equation gives: 
\begin{align*}
 \left(\sum_{k \geq 2} k \cdot r_k\right) + (|\chi|-r)  - 1 &= |\chi|-1 + d_{\mathrm{MP}}(T_1,T_2), \\
  \left(\sum_{k \geq 2} k \cdot r_k\right) - r &=  d_{\mathrm{MP}}(T_1,T_2), \\
\sum_{k \geq 2} (k-1)r_k& = d_{\mathrm{MP}}(T_1,T_2),
\end{align*}
where the last equality is due to $r=\sum_{k\geq 2} r_k$ by definition. The theorem now follows from Equations \eqref{eq:prior} and \eqref{eq:crucial}.
\end{proof}

\clearpage

\section{A lower bound on the number of states}
\label{sec:mareike}

In the proof of Lemma 3.7 of \cite{FischerKelk2016}, it was shown that there are two trees $T_1$ and $T_2$ such that $d_{\mathrm{MP}}(T_1,T_2)=2$, but the only way to reach this distance is by employing a convex character of at least 3 states (or by employing a 2-state character that is not convex on either one of the trees). 

We now generalize this finding with the following observation.

\begin{prop}\label{prop:dmpStatesNotEnough}
Let $r \geq 1$. Then, there exist two binary phylogenetic trees $T_1, T_2$ on $n=4r$ leaves such that $d_{\mathrm{MP}}(T_1,T_2)= r$ and such that \begin{enumerate} 
\item$ \vert \ell(T_1,\chi)-\ell(T_2,\chi)\vert=r$ for some character $\chi$ on precisely $r+1$ states which is convex on $T_1$ or $T_2$, and

\item $\max\limits_{\chi}\vert \ell(T_1,\chi)-\ell(T_2,\chi)\vert<r$ if the maximum runs over all characters $\chi$ with up to $r$ states that are convex on either $T_1$ or $T_2$.
\end{enumerate}
\end{prop}
\begin{proof}
We give an explicit construction of two trees $T_1$ and $T_2$ which fulfill both parts 1 and 2 of the proposition. For convenience and in order to use arguments based on Fitch's algorithm, we will construct them as rooted trees, but the statements hold analogously when the root position is disregarded. 

\begin{itemize} \item \textbf{Tree construction:} Both trees $T_1$ and $T_2$ which we construct in the following can be thought of as \enquote{caterpillars} of \enquote{double cherries}. 
In particular, we take the two starting trees $S_1= ((1,2),(3,4))$ and $S_2= ((1,3),(2,4))$ where $((a,b),(c,d))$ denotes the unique  rooted binary phylogenetic tree with two cherries $\{a,b\}$ and $\{c,d\}$. For $r=1$, we set $T_1=S_1$ and $T_2=S_2$. For $r\geq 2$, however, we start by taking $r$ copies of $S_1$ for $T_1$ and $r$ copies of $S_2$ for $T_2$. Then, we change the leaf labels in each of these copies as follows: In the $k^{th}$ copy of $S_1$ or $S_2$ (for $1\leq k \leq r$), which we will later on refer to as $S_k^1$ or $S_k^2$, respectively, we replace each leaf label $x$ by $x+4\cdot (k-1)$. This way, the leaf labels of the first copy remain unchanged (as $x+4\cdot (1-1)=x$), the second copy gets leaf labels $5=1+4\cdot (2-1)$ up to $8=4+4\cdot (2-1)$ and so forth. Last, for both $T_1$ and $T_2$, we take a caterpillar with $r$ leaves such that the unique cherry has labels 1 and 2 and the other leaves are subsequently labeled $3, \ldots, r$, starting with the leaf closest to the cherry and ending with the leaf closest to the root, cf. Figure \ref{fig:catbackbone}. Then we identify the $k^{th}$ leaf of this caterpillar with the root of the $k^{th}$ copy of $S_1$ to construct $T_1$ and with the root of the $k^{th}$ copy of $S_2$ to construct $T_2$. We refer to the caterpillar tree to which the copies of $S_1$ or $S_2$, respectively, have been assigned, as the \enquote{caterpillar backbone} of $T_1$ or $T_2$, respectively. An illustration for the case $r=3$ can be seen in Figure \ref{fig:catbackbone}. 

\begin{figure}[ht]      \centering
    \includegraphics[width=\linewidth]{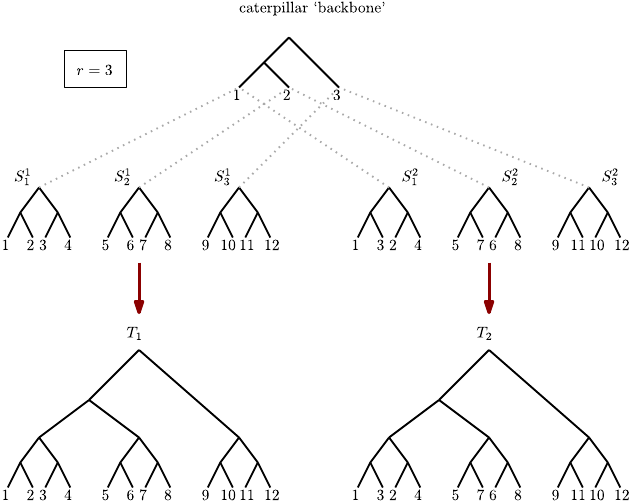}
    \caption{The construction of $T_1$ and $T_2$ as used in the proof of Proposition \ref{prop:dmpStatesNotEnough}. } \label{fig:catbackbone}
  \end{figure}

\item We now show that $d_{\mathrm{MP}}(T_1,T_2)\geq r$. In order to see this, consider the following character $\chi$: for $k=1,\ldots,r$, it assigns leaves with number $1+4(k-1)$ and $2+4(k-1)$ state $A$, i.e., leaves 1, 2, 5, 6, 9, 10, etc. get assigned state $A$. Moreover, it assigns to all remaining leaves in subtree $S_k^1$ (and hence $S_k^2$) state $B_k$ (for $k=1,\ldots,r$) -- so in total, $\chi$ employs $r+1$ states. It can be easily seen that $\ell(\chi,T_2)=2r$, which is due to the fact that $T_2$ has $2r$ cherries, all of which use two states.  However, as $A$ is contained in the Fitch set $\Phi(\rho_k^2)$ of the root of $S_k^2$ for each of the $r$ subtrees $S_k^2$ ($k=1,\ldots,r$) in $T_2$, for the caterpillar backbone of the tree, no more unions need to be taken. Thus, we indeed have $\ell(\chi,T_2)=2r$.

Similarly, it can easily be seen that $\ell(\chi,T_1)=r$, which is due to the fact that $T_1$ has $2r$ cherries, each of which uses a single state with a \enquote{sibling-cherry} in a different single state. Thus, in this case, each of the $r$ subtrees $S_k^1$ in $T_1$ requires one change each. However, as $A$ is contained in the Fitch set $\Phi(\rho_k^1)$ of the roots of the subtrees $S_k^1$ in $T_1$, for the caterpillar backbone of the tree, like with $T_2$, no more unions need to be taken. Thus, we have $\ell(\chi,T_2)-\ell(\chi,T_1)=2r-r=r$, which shows that $d_{\mathrm{MP}}(T_1,T_2) \geq r$.

\item Next we show that $d_{\mathrm{MP}}(T_1,T_2)\leq r$. In order to do so, let $\chi$ be a character with $|\ell(T_1,\chi)-\ell(T_2,\chi)|=d_{\mathrm{MP}}(T_1,T_2)$ and which is convex on one of the two trees and which, amongst all characters with these two properties, has a minimum number of states.
As $T_1$ and $T_2$ only differ in their leaf labelling, we may assume by symmetry and without loss of generality that $\chi$ is convex on $T_1$, i.e., we may drop the absolute value and assume $\ell(T_2,\chi)-\ell(T_1,\chi)=d_{\mathrm{MP}}(T_1,T_2)$. As $\chi$ employs the minimum number of states, we can derive from Observation \ref{observation:split} that it does not contain any taxa that are assigned a unique state. In order to see this, assume the contrary, i.e., there is a state, say $A$,  which is only assigned to one leaf, say $x$, by $\chi$. Let $B$ denote the state assigned to the cherry sibling, say $y$, of $x$ in $T_1$. Then the Fitch set of the parent $v$ of $x$ and $y$ will be $\Phi(v)=\{A,B\}$, which shows that $A$ and $B$ will be adjacent in $F_1$ regardless of the
most parsimonious extension of
$\chi$ we choose. Moreover, the pending edge leading to $x$ in $T_2$ induces a split separating all leaves in state $B$ from the single leaf $x$ in state $A$. This shows that states $A$ and $B$ meet the requirements of Observation \ref{observation:split}, so they form a good pair. This, however, contradicts  the assumption that $\chi$ was chosen as an optimal character with the minimum number of states. Thus, $\chi$ must employ each state at least twice\footnote{That we can always assume the existence of an optimal convex character in which each state is employed at least twice, was originally proven for all pairs of input trees in \cite[Observation 6.1]{kelk2017complexity}. We have given details here to make the current proof more self-contained.}. 

Next, observe that for all $k=1,\ldots,r$, we must have that there are at most two states assigned by $\chi$ to the leaves of $S_k^1$ in $T_1$ (and thus, automatically, also to the leaves of $S_k^2$ in $T_2$). This is because, if some $S_k^1$ uses three or more states, then from the convexity of $\chi$ on
$T_1$ it follows that at least one of these states will occur on only one leaf of $T_1$, contradicting the fact that each state is employed at least twice. Similarly, it cannot happen that there are two states assigned by $\chi$ to the leaves of $S_k^1$, such that one of the states occurs on three of the four leaves in $S_k^1$
and the other state on only one of the four leaves: again, by convexity, the state that appears on only one leaf of $S_k^1$ would appear on only one leaf in $T_1$.

So we may indeed conclude that in  each subtree $S_k^1$ of $T_1$, $\chi$ employs at most two states -- and if there are precisely two states contained in a subtree  $S_k^1$ (and $S_k^2$), both states are employed  twice. 

The latter observation has an important implication. Let us consider leaves $\{1,2,3,4\}$, for instance. We write
$\chi\!\left.\vphantom{\big|}\right|_{\{1,2,3,4\}}$ to denote the restriction of $\chi$ to the taxa $\{1,2,3,4\}$. By the convexity of $\chi$ on $T_1$, the only two possibilities 
are $\chi\!\left.\vphantom{\big|}\right|_{\{1,2,3,4\}}=AAAA$ and $\chi\!\left.\vphantom{\big|}\right|_{\{1,2,3,4\}}=AABB$. Crucially, regardless of which holds, the Fitch sets assigned to the roots $\rho_1^1$ of $S_1^1$ and $\rho_1^2$ of $S_1^2$ are identical, even though this set is sometimes the result of taking the union and sometimes of an intersection. The same is true for all subtrees $S_k^1$ and $S_k^2$ (for $k=1,\ldots,r$) by symmetry.

The fact that for $\chi$ we have $\Phi(\rho_k^1)=\Phi(\rho_k^2)$ for all $k=1,\ldots,r$ implies that the bottom-up phase of Fitch's algorithm will assign precisely the same number of union sets to the caterpillar backbone of $T_1$ as to the caterpillar backbone of $T_2$. This, in turn, implies that the contribution of these sets to the difference $\ell(\chi,T_2)-\ell(\chi,T_1)$ is 0 (as they cancel out). Thus, only the number of union sets assigned to the inner nodes of the subtrees $S_k^1$ and $S_k^2$ contributes to the difference $\ell(\chi,T_2)-\ell(\chi,T_1)$ and thus to $d_{\mathrm{MP}}(T_1,T_2)$.

Now, if in $S_k^1$ and $S_k^2$ only one character state is used, i.e., if $\chi$ restricted to the taxon set of this subtree is, say, $AAAA$, then the parsimony score on both subtrees is 0, and so is the difference between these scores. Thus, subtrees $S_k^1$ and $S_k^2$ in which only one state is used do not contribute to $d_{\mathrm{MP}}(T_1,T_2)$. If, however, two states are used twice, then by the convexity of $\chi$ on $T_1$ the only possibility is $AABB$. In this case, we have score 1 on $S_k^1$ and score 2 on $S_k^2$, so the contribution of subtrees $S_k^1$ and $S_k^2$ to the difference $\ell(\chi,T_2)-\ell(\chi,T_1)$ is 1.

In summary, the caterpillar backbone does not contribute to $d_{\mathrm{MP}}(T_1,T_2)$, only subtrees $S_k^1$ and $S_k^2$ do (for $k=1,\ldots,r$). However, each pair $(S_k^1,S_k^2$) of these subtrees contributes at most 1 to $d_{\mathrm{MP}}(T_1,T_2)=\ell(\chi,T_2)-\ell(\chi,T_1)$, and as we have $r$ such pairs, this shows that $d_{\mathrm{MP}}(T_1,T_2)\leq r$.

Together with the previous part of the proof, this shows that $d_{\mathrm{MP}}(T_1,T_2)=r$, and we have also already seen above that this distance can be achieved with a character employing precisely $r+1$ states. So we know that part 1 of the proposition holds.

\item It remains to establish part 2 of the proposition. That is, to  show that $|\ell(T_2,\chi)-\ell(T_1,\chi)|<r$ for all characters $\chi$ which are convex on $T_1$ (or $T_2$, but due to symmetry, we do not need to consider this for our construction) but employ at most $r$ states. Let $\chi$ in the following be again an optimal character with the minimum number of states, such that $\chi$ is convex on $T_1$, and assume that we have $d_{\mathrm{MP}}(T_1,T_2)=\ell(\chi,T_2)-\ell(\chi,T_1)=r$. 

Recall the above arguments with which we showed that an optimal character $\chi$ with the minimum number of states must assign either only one state to the leaves of each subtree $S_k^1$ (and thus also $S_k^2$), or two states twice. 

Suppose there is at least one $i \in \{1,\ldots,r\}$ such that $\chi$ assigns only one state, say $A$, to $S_i^1$. Then, as we have also seen above, the parsimony scores of the restriction of $\chi$ to the taxa of $S_i^1$  are 0 both for 
$S_i^1$ and $S_{i}^2$. As the caterpillar backbone does not contribute to $d_{\mathrm{MP}}$ but only the $S_k^1$ and $S_k^2$ subtrees do (again as shown above), there are now, due to the existence of $i$, at most $r-1$ values of $k$ such that $S_k^1$ and $S_k^2$ contribute to $d_{\mathrm{MP}}$. As shown above, each pair $(S_k^1,S_k^2)$ of such subtrees can contribute at most 1 to $d_{\mathrm{MP}}(T_1,T_2)$, which shows that $d_{\mathrm{MP}}(T_1,T_2)\leq r-1$. This is a contradiction, as we have already shown that $d_{\mathrm{MP}}(T_1,T_2)=r$. 

Therefore, we know that none of the subtrees $S_k^1$ can employ only one state. So all of these subtrees have to employ precisely two states twice. 

Now, observe by the convexity of $\chi$ on $T_1$ that there cannot be a subtree $S_k^1$ such that the two states, say $A$ and $B$, assigned to this subtree by $\chi$ are \emph{both} also present in another such subtree.

Therefore, for each $k\in\{1,\ldots,r\}$, $S_k^1$ contains at most one state that also occurs in at least one more subtree, say $S_j^1$ (with $k\neq j$). Thus, for each $k \in \{1,\ldots,r\}$ we have one state used by $\chi$ which is unique to subtree $S_k^1$. This leads to $r$ different states which label precisely half of the leaves of $T_1$. The leaves that are then still unlabelled altogether require at least one more state, so $\chi$ employs at least $r+1$ states. This shows that every character obtaining $d_{\mathrm{MP}}(T_1,T_2)$ must employ at least $r+1$ states. 
\end{itemize}
This completes the proof.
\end{proof}

\section{Empirical analysis}
\label{sec:experiment}
To better understand the significance of our bounds, we performed the following simple experiment. We took the 735 subtree-reduced\footnote{In this case, subtree-reduced refers to collapsing common pendant subtrees, i.e. repeatedly collapsing common cherries into a single leaf until there are no more common cherries. Note that this kind of reduction does not alter $d_{\mathrm{MP}}$ or $d_{\mathrm{TBR}}$ \cite{allen2001subtree,kelk2016reduction}. It also does not affect the minimum number of states required by convex characters achieving $d_{\mathrm{MP}}$.} tree pairs from the dataset in \cite{van2022reflections} (which were further analyzed in \cite{FROHN2025107364}) and extracted those for which $d_{\mathrm{MP}}$ is known exactly: these are tree pairs for which known lower bounds on $d_{\mathrm{MP}}$ match known upper bounds on $d_{\mathrm{TBR}}$ (recall that $d_{\mathrm{MP}} \leq 
d_{\mathrm{TBR}}$).
There were 644 such tree pairs. The first two columns in Table \ref{tab:stats} below show the minimum, maximum, average and standard deviation of the number of taxa and $d_{\mathrm{MP}}$ amongst these 644 trees. For each tree pair we repeatedly randomly sampled characters that are convex on one of the trees, noting the minimum number of states amongst all those convex characters that achieved a distance of $d_{\mathrm{MP}}$; for this we used the efficient random sampling strategy  described in \cite{kelk2017note}. Due to the fact that we were only sampling, this is only an \emph{upper bound} $\hat{s}$ on the minimum number of states $s$ required for the given tree pair. Summary statistics for $\hat{s}$ are shown in the third column of the table and for $\frac{\hat{s}}{ d_{\mathrm{MP}}}$ in the fourth column. Recall that if $\hat{s} = s$ i.e. $\hat{s}$ is the true minimum number of states, this ratio will be at most 2, due to the main theorem of this paper.
The fifth column summarizes the additive gap between the $d_{\mathrm{MP}}+1$ lower bound and $\hat{s}$.

Clearly, there is no evidence in this experiment that the $d_{\mathrm{MP}}+1$ lower bound must be strengthened (i.e. increased) further, since the MIN value in the fifth column is non-negative. Given that tree pairs requiring $d_{\mathrm{MP}}+1$ states had to be very carefully constructed (see Section \ref{sec:mareike}), it is not altogether
surprising that all the tree pairs in the experiment required at most (in fact, strictly fewer than) $d_{\mathrm{MP}}+1$ states. What is somewhat surprising is 
the average bound of only $0.44 d_{\mathrm{MP}}$ states, and the fairly tight distribution around
this mean (a standard deviation of 0.16). This suggests that for many trees far fewer than $d_{\mathrm{MP}}$ states are required by
convex characters to achieve $d_{\mathrm{MP}}$; this is a long way from our proven upper bound of $2d_{\mathrm{MP}}$. Given that the number of convex characters grows comparatively
slowly, this also suggests that for quite high $d_{\mathrm{MP}}$ an enumeration or sampling strategy amongst the space of convex characters with at most $0.44d_{\mathrm{MP}}$ states often suffices to discover $d_{\mathrm{MP}}$, and is computationally not too taxing, although discovering such a convex character does not constitute a certificate: this also requires a matching upper bound on $d_{\mathrm{MP}}$. It is well-known that in many cases (upper bounds on) $d_{\mathrm{TBR}}$ can function
as such a matching bound, but in cases where the two distances do not match there is still very little known about how to certify the optimality of a given (convex) character. Indeed, the possibility exists that the average
bound of $0.44 d_{\mathrm{MP}}$ is influenced by the fact that all tree pairs in this experiment have the
property $d_{\mathrm{MP}} = d_{\mathrm{TBR}}$, although it is currently computationally 
infeasible to conduct
such an experiment without this assumption.

The full data can be downloaded from \url{https://github.com/skelk2001/boundingconvexstates/}.
\begin{table}[h!]
\centering
\begin{tabular}{lccccc}
\hline
 & taxa & $d_{\mathrm{MP}}$ & $\hat{s}$ & $\frac{\hat{s}}{d_{\mathrm{MP}}}$ & ($d_{\mathrm{MP}}$+1)-$\hat{s}$ \\
\hline
MIN & 18 & 4 & 2 & 0.17 & 1 \\
AVG & 94.50 & 17.87 & 7.04 & 0.44 & 11.83 \\
MAX & 223 & 35 & 17 & 1.00 & 28 \\
STDEV & 42.68 & 9.18 & 3.22 & 0.16 & 6.88 \\
\hline
\end{tabular}
\caption{Summary statistics for (1) the number of taxa, (2) $d_{\mathrm{MP}}$, (3) the computed upper bound $\hat{s}$ on the number of states
required by convex characters achieving $d_{\mathrm{MP}}$, (4) the computed upper bound divided by $d_{\mathrm{MP}}$, (5) 
the additive gap between the lower bound of $d_{\mathrm{MP}}+1$ states and the computed upper bound. The dataset consists of 644 tree pairs for which $d_{\mathrm{MP}}$ is known. Fractional numbers are rounded to two decimal places.
}
\label{tab:stats}
\end{table}

\section{Conclusions}
\label{sec:future}

We have shown that there always exists an optimal convex character with at most $2d_{\mathrm{MP}}(T_1, T_2)$ states, improving upon the previous bound of $7d_{\mathrm{MP}}(T_1, T_2)-5$. We conjecture that our bound can be further strengthened to $d_{\mathrm{MP}}(T_1, T_2)+1$ states, which would then match our lower bound. We have several reasons for believing this. First, the unique states bound inherent in Theorem \ref{thm:moonshot}, beyond which good pairs are guaranteed to exist, can probably be improved for repeating states with multiplicity 3 or higher. However, it becomes increasingly difficult and intricate to argue that in $T_2$ the parsimony score of the relabelled character $\newchi$ does not decrease too much compared to the original character $\chi$. Also, a strengthening of Theorem \ref{thm:moonshot} does not automatically improve the overall bound on the number of states required in an optimal convex character. A second reason for optimism is that our technique currently only considers \emph{unique} states that are adjacent in $T_1$ to a \emph{single} repeating state $B$. It is plausible that both unique and repeating states can function as blocking states
in a wider array of situations than presently used. However, the technicalities involved will indeed be considerable. By way of illustration, we have a sketch proof, not shown here, that for $d_{\mathrm{MP}}(T_1, T_2)=2$, 3 states are sufficient, but the techniques involved in the proof and the resulting details are rather specific and are not yet generalisable. Third, although only a purely empirical insight, the experiment in Section \ref{sec:experiment} hints that it is far more likely that our upper bound needs to be reduced significantly, than that the lower bound needs to be increased.

Finally, we return to the result \cite{jones2021maximum} where following \cite{kelk2016reduction} it is shown that, after applying the subtree and chain reduction rules to exhaustion, two new trees with the same $d_{\mathrm{MP}}$ are obtained with (ignoring additive terms) at most $560 d_{\mathrm{MP}}$ taxa. Our $2d_{\mathrm{MP}}$ upper bound on the number of states has a similar function-of-$d_{\mathrm{MP}}$ flavour that is characteristic of the kernelization literature \cite{fomin2019kernelization}. It is therefore quite natural to ask whether our upper bound, or the machinery we have developed to prove it, could be used to lower the constant 560. It is clear that a convex character with at most $2d_{\mathrm{MP}}$ states on a tree with perhaps as many as $560 d_{\mathrm{MP}}$ taxa will always have at least one character state that contains `very many' taxa. Perhaps the regions of the two trees that have this state can be shrunk without altering the parsimony score on either tree, thus reducing the number of taxa overall and thus hopefully also the constant 560. However, it is currently far from obvious how a polynomial-time data reduction rule could be designed to target such structures; this requires further research.

\bibliographystyle{plain} 
\bibliography{References}

\end{document}